\documentclass[11pt,a4paper,english]{article}

\usepackage{amssymb}
\usepackage{amsmath}
\usepackage{epsfig}
\usepackage{graphicx}
\usepackage{color}
\usepackage{bigints}
\usepackage{array}
\usepackage[normalem]{ulem}
\usepackage{xspace}
\usepackage{hyperref}

\usepackage{url}
\usepackage{amsmath}

\graphicspath{{./}{figures/}}

\newcommand{\up}{\mathit{up}}
\newcommand{\low}{\mathit{low}}
\newcommand{\reg}{\mathit{reg}}
\newcommand{\slab}{\mathit{slab}}
\newcommand{\gaplist}{\mathit{gaplist}}

\newcommand{\wt}{\mathit{wt}}

\newcommand{\eps} {\varepsilon}

\newcommand{\real}	{\mathbb{R}}

\newcommand{\cancel}[1] {}

\newenvironment{romani}
{
	
	\begin{enumerate}}
	{\end{enumerate}}

\newtheorem{theorem}{Theorem}[section]
\newtheorem{lemma}{Lemma}[section]
\newtheorem{corollary}{Corollary}[section]

\newbox\ProofSym
\setbox\ProofSym=\hbox{%
	\unitlength=0.18ex%
	\begin{picture}(10,10)
	\put(0,0){\framebox(9,9){}}
	\put(0,3){\framebox(6,6){}}
	\end{picture}}

\begin{document}

\title{Dynamic Distribution-Sensitive Point Location\thanks{Supported by Research Grants Council, Hong Kong, China (project no.~16201116).  An extended abstract appears in Proceedings of the 36th International Symposium on Computational Geometry, 2020.}}

\author{Siu-Wing Cheng\footnote{Department of Computer Science and Engineering, HKUST, Hong Kong.  Email : {\tt scheng@cse.ust.hk, lmkaa@connect.ust.hk}} \and Man-Kit Lau\footnotemark[2]}

\date{}

\maketitle

\begin{abstract}
	We propose a dynamic data structure for the distribution-sensitive point location problem.  Suppose that there is a fixed query distribution in $\real^2$, and we are given an oracle that can return in $O(1)$ time the probability of a query point falling into a polygonal region of constant complexity.   We can maintain a convex subdivision $\cal S$ with $n$ vertices such that each query is answered in $O(\mathrm{OPT})$ expected time, where OPT is the minimum expected time of the best linear decision tree for point location in $\cal S$.  The space and construction time are $O(n\log^2 n)$.  An update of $\cal S$ as a mixed sequence of $k$ edge insertions and deletions takes $O(k\log^5 n)$ amortized time.  As a corollary, the randomized incremental construction of the Voronoi diagram of $n$ sites can be performed in $O(n\log^5 n)$ expected time so that, during the incremental construction, a nearest neighbor query at any time can be answered optimally with respect to the intermediate Voronoi diagram at that time.
\end{abstract}

\section{Introduction}

Planar point location is a classical problem in computational geometry.  
In the static case, a subdivision is preprocessed into a data structure so that, given a query point, the face containing it can be reported efficiently.  In the dynamic case, the data structure needs to accommodate edge insertions and deletions.   It is assumed that every new edge inserted does not cross any existing edge.  There are well-known worst-case optimal results in the static case~\cite{paper:Adamy1998,paper:Edelsbrunner1986,paper:Kirkpatrick1981,paper:Sarnak1986}.  There has been a long series of results in the dynamic case~\cite{paper:Arge2006,paper:Baumgarten1994,paper:Chan2015,paper:Cheng1992,CPT96,paper:Chiang1992,paper:Goodrich1998,O18,OA18}.
%
For a dynamic connected subdivision of $n$ vertices, an $O(\log n)$ query time and an $O(\log^{1+\eps} n)$ update time for any $\eps > 0$ can be achieved~\cite{paper:Chan2015}.

When the faces have different probabilities of containing the query point, minimizing the expected query time is a more appropriate goal.  Assume that these probabilities are given or accessible via an oracle.  Arya et~al.~\cite{paper:Arya2007b} and 
Iacono~\cite{paper:Iacono2004} obtained optimal expected query time when the faces have constant complexities.
Later, Collete~et~al.~\cite{paper:Collette2012} obtained the same result for connected subdivisions.  So did 
Afshani~et~al.~\cite{paper:Afshani2015} and Bose et al.~\cite{paper:Bose2013} for general subdivisions.

In the case that no prior information about the queries is available, Iacono and Mulzer~\cite{paper:Iacono2011} designed a method for triangulations that can process an online query sequence $\sigma$ in time proportional to $n$ plus the entropy of $\sigma$.
We developed solutions for convex and connected subdivisions in a series of work~\cite{cheng2015adaptive,paper:cheng2017adaptive,CL18}.  For convex subdivisions, the processing time is $O(T_{\text{opt}} + n)$, where $T_{\text{opt}}$ is the minimum time needed by a linear decision tree to process~$\sigma$~\cite{paper:cheng2017adaptive}.  For connected subdivisions, the processing time is $O(T_{\text{opt}} + n + |\sigma|\log(\log^* n))$~\cite{CL18}.

In this paper, we are interested in dynamic distribution-sensitive planar point location.  Such a problem arises when there are online demands for servers that open and close over time, and a nearest server needs to be located for a demand.  For example, walking tourists may look for a facility nearby (e.g.~convenience store) and search on their mobile phones.  The query distribution can be characterized using historical data.  New convenience store may open and existing ones may go out of business.  If we use the Euclidean metric, then we are locating a query point in a dynamic convex subdivision which is a Voronoi diagram.  We are interested in solutions with optimal expected query time.  

We assume that there is an oracle that can return in $O(1)$ time the probability of a query point falling inside a polygonal region of constant complexity.  We propose a data structure for maintaining a convex subdivision $\cal S$ with $n$ vertices such that each query is answered in $O(\mathrm{OPT})$ expected time, where OPT is the minimum expected time of the best \emph{point location decision tree for $\cal S$}, i.e., the best linear decision tree for answering point location queries in $\cal S$ with respect to the fixed underlying query distribution.  An update of $\cal S$ as a mixed sequence of $k$ edge insertions and deletions can be performed in $O(k\log^5 n)$ amortized time.  The space and construction time are $O(n\log^2 n)$.   As a corollary, we can carry out the randomized incremental construction of the Voronoi diagram of $n$ sites so that, during the incremental construction, a nearest neighbor query at any time can be answered optimally with respect to the intermediate Voronoi diagram at that time.  The expected total construction time is $O(n\log^5 n)$ because each site insertion incurs $O(1)$ expected structural changes to the Voronoi diagram.  A key ingredient in our solution is a new data structure, \emph{slab tree}, for maintaining a triangulation with a nearly optimal expected point location time and polylogarithmic amortized update time.  We believe that this data structuring technique is of independent interest and it may find other applications, especially in a distribution-sensitive setting.

\section{Overview}

There are two aspects of the maintenance of a convex subdivision $\cal S$ for point location.  First, the maintenance of $\cal S$ and a decomposition of $\cal S$ into simpler shapes in which the point location is carried out.  Second, the maintenance of the point location structure.

The maintenance of a \emph{balanced geodesic triangulation} of a connected planar subdivision has been studied by Goodrich and Tamassia~\cite{paper:Goodrich1997}.  It is shown that every edge update in a planar subdivision can be transformed to $O(\log n)$ edge updates in its balanced geodesic triangulation.  This method fits nicely with our previous use of the DK-triangulation~\cite{Dobkin1990} of $\cal S$ in adaptive point location~\cite{paper:cheng2017adaptive} because a DK-triangulation is a balanced geodesic triangulation.   The preamble of Section~\ref{sec:dynamconv} defines a convex subdivision and the updates to be supported.  Section~\ref{sec:dk} defines the DK-triangulation of $\cal S$ and the performance of Goodrich and Tamassia's structure in our case.

The development of a dynamic distribution-sensitive point location structure for the DK-triangulation of $\cal S$ is the main thrust of this paper.  Sections~\ref{sec:dyntridatastructure} and~\ref{sec:update} are devoted to it.  Theorem~\ref{thm:3} in Section~\ref{sec:general} gives the performance of this dynamic data structure.   Query time is $O(\mathrm{OPT} + \log\log n)$ expected and update time is $O(k\log^4 n)$ amortized, where OPT is the minimum expected query time of the best point location decision tree, $n$ is the number of vertices of $\cal S$, and $k$ is the number of edge updates involved.  In Sections~\ref{sec:pt} and~\ref{subsec:convplsupdate}, we discuss how to apply Theorem~\ref{thm:3} to obtain the main result of this paper, Theorem~\ref{thm:main} in Section~\ref{subsec:convplsupdate}, on dynamic distribution-sensitive point location.   Then, the result on answering queries optimally during the randomized incremental construction of the Voronoi diagram of $n$ points follows as Corollary~\ref{cor:main} in Section~\ref{subsec:convplsupdate}.  Since the expected query time in Theorem~\ref{thm:3} has an $O(\log\log n)$ additive term, more work is needed in applying this result in Section~\ref{sec:pt} in order to obtain an optimal query time.  This is achieved by adapting our previous work~\cite{paper:cheng2017adaptive} to the distribution-sensitive setting.  

The dynamic point location structure for the DK-triangulation of $\cal S$ is developed in three stages.  First, we describe in Section~\ref{sec:dyntridatastructure} the \emph{slab tree} for performing point location in $O(\mathrm{OPT} + \log\log n)$ time in a static triangulation.   Here, we assume that a fixed set $\cal L$ of vertical lines is given such that the vertices of $\cal S$ lie on lines in $\cal L$, but some lines in $\cal L$ may not pass through any vertex of $\cal S$.   This feature is very useful when accommodating updates.   

Section~\ref{sec:update} defines a \emph{triangulation-update} and describes how to perform it in the special case that any new vertex that will appear must lie on some line in a fixed, given set $\cal L$.  It is based on a recursive traversal of the slab tree that performs a merge of the updated portions of the triangulation with the existing information stored at every node visited. The inductive proof of the correctness of this merging process is quite involved, so it is deferred to the appendix.   Lemma~\label{lem:1} in Section~\ref{sec:update} summarizes the performance of this semi-dynamic structure.  

Finally, Section~\ref{sec:general} discusses how to accommodate arbitrary vertex location.  This is achieved by generalizing the slab tree so that: (1)~each internal node has a fan-out of $O(\log n)$ instead of three in the static and semi-dynamic cases, (2)~the children of an internal nodes are classified as light or heavy based on their probabilities of containing a query point, and~(3)~the heavy and light classification will allow us to periodically choose appropriate slab subtrees to be rebuilt.  The above generalization allows vertices at arbitrary locations to be inserted in polylogarithmic amortized time.

\section{Dynamic convex subdivision}
\label{sec:dynamconv}

Let $\cal S$ be a convex subdivision.   Let $\partial {\cal S}$ be the outer boundary of $\cal S$, which bounds a convex polygon.
A \emph{general-update sequence} $\Phi$ is a mixed sequence of edge insertions and deletions in $\cal S$ that produces a convex subdivision.
However, the intermediate subdivision after each edge update  is only required to be connected, not necessarily convex.   Vertices may be inserted into or deleted from $\partial {\cal S}$, but the shape of $\partial {\cal S}$ is never altered.   We will present in Sections~\ref{sec:dyntridatastructure}-\ref{sec:general} a dynamic point location structure for a DK-triangulation of $\cal S$ (to be defined below).  
The performance of this structure is given in Theorem~\ref{thm:3}.   In this section, we show how to apply Theorem~\ref{thm:3} to obtain a dynamic distribution-sensitive point location structure for $\cal S$.

\subsection{Dynamic DK-triangulation}
\label{sec:dk}

Let $P$ be a convex polygon.  Find three vertices $x$, $y$ and $z$ that roughly trisect the boundary of $P$.  This gives a triangle $xyz$.  Next, find a vertex $w$ that roughly bisects the chain delimited by $x$ and $y$.  This gives a triangle $xyw$ adjacent to $xyz$.  We recurse on the other chains to produce a DK-triangulation of $P$~\cite{Dobkin1990}.  It has the property that any line segment inside $P$ intersects $O(\log |P|)$ triangles.  A DK-triangulation of $\cal S$ is obtained by computing the DK-triangulations of its bounded faces.
%
Goodrich and Tamassia~\cite{paper:Goodrich1997} proposed a method to maintain a \emph{balanced geodesic triangulation} of a connected subdivision.  
We can use it to maintain a DK-triangulation of $\cal S$ because a DK-triangulation is a balanced geodesic triangulation.  By their method, each edge insertion/deletion in $\cal S$ is transformed into $O(\log n)$ edge insertions and deletions in the DK-triangulation of $\cal S$, where $n$ is the number of vertices of $\cal S$.  Consequently, each edge insertion/deletion in $\cal S$ takes $O(\log^2 n)$ time.

\subsection{Point location}
\label{sec:pt}


We modify our adaptive point location structure for static convex subdivisions~\cite{paper:cheng2017adaptive} to make it work for the distribution-sensitive setting.  Compute a DK-triangulation $\Delta_1$ of $\cal S$.  For each triangle $t \in \Delta_1$, use the oracle to compute the probability $\Pr(t)$ of a query point falling into $t$.  This probability is the weight of that triangle.  We call the triangles in $\Delta_1$ \emph{non-dummy} because we will introduce some \emph{dummy} triangles later.

Construct a data structure $D_1$ for $\Delta_1$ with two parts.  The first part of $D_1$ is a new dynamic distribution-sensitve point location structure for triangulations (Theorem~\ref{thm:3}).  The query time of the first part of $D_1$ is $O(\text{OPT}+ \log\log n)$, where OPT is the minimum expected time of the best point location decision tree for $\Delta_1$. The second part can be any dynamic point location structure with $O(\log n)$ query time, provided that its update time is $O(\log^2 n)$ and space is $O(n\log^2 n)$~\cite{paper:Arge2006, paper:Chan2015, paper:Chiang1992, paper:Preparata1989}. 

Next, we build a hierarchy of triangulations and their point location structures from $\Delta_1$.  The triangulation size drops exponentially from one level to the next, by promoting a polylogarithmic number of triangles with the highest probabilities of containing a query point.  The hierarchy serves as a multi-level cache.  A point location query will start from the highest level and work downward until the query is answered.  This results in the optimal expected query time as given in Lemma~\ref{lem:dsexpected} below, the proof of which is an adaptation of our previous result in~\cite{paper:cheng2017adaptive,CL18}.

Specifically, for $i \ge 2$, define $n_i = (\log_2 n_{i-1})^4$ inductively, where $n_1 = n$.  To construct $\Delta_i$ from $\Delta_{i-1}$, extract the non-dummy triangles in $\Delta_{i-1}$ whose probabilities of containing a query point are among the top $(\log_2 n_{i-1})^4$.  For each subset of extracted triangles that lie inside the same bounded face of $\cal S$, compute their convex hull and its DK-triangulation.   These convex hulls are holes  in the polygon $H_i$ with $\partial {\cal S}$ as its outer boundary.  Triangulate $H_i$.  We call the triangles used in triangulating $H_i$ \emph{dummy} and the triangles in the DK-triangulations of the holes of $H_i$ \emph{non-dummy}.   The dummy and non-dummy triangles form the triangulation $\Delta_i$.  The size of $\Delta_i$ is $O(n_i)$.   For each non-dummy triangle $t \in \Delta_i$, set its weight to be $\max\bigl\{\Pr(t), W^*_i/n_i \bigr\}$, where $W^*_i$ is the sum of $\Pr(t)$ over all non-dummy triangles $t$ in $\Delta_i$.  Dummy triangles are given weight $W^*_i/n_i$.    Note that the total weight $W_i$ of all triangles in $\Delta_i$ is $\Theta(W^*_i)$.  Construct a data structure $D_i$ 
as the point location structure of Iacono~\cite{paper:Iacono2004} for $\Delta_i$, which can answer a query in $O\bigl(\log \frac{W_i}{w_i}\bigr)$ time, where $w_i$ is the weight of the triangle containing the query point.
The query time of $D_i$ is no worse than $O(\log n_i)$ in the worst case as $w_i \geq W_i^*/n_i = \Theta(W_i/n_i)$.

A hierarchy $(\Delta_1,D_1), \ldots, (\Delta_m,D_m)$ is obtained in the end, where the size of $\Delta_m$ is less than some predefined constant.  So $m = O(\log^* n)$.

For $i \geq 2$, label every non-dummy triangle $t \in \Delta_i$ with the id of the bounded face of $\cal S$ that contains it.  If $t$ is located by a query, we can report the corresponding face of $\cal S$.   The labelling of triangles in $\Delta_1$ is done differently in order to allow updates in $\Delta_1$ to be performed efficiently.  For each vertex $p$ of $\cal S$, its incident triangles in $\Delta_1$ are divided into circularly consecutive groups by the incident edges of $p$ in $\cal S$.  Thus, each group lies in a distinct face of $\cal S$ incident to $p$.  We store these groups in clockwise order in a biased search tree $T_p$~\cite{paper:bent85} associated with $p$.  Each group is labelled by the bounded face of $\cal S$ that contains it.  The group weight is the maximum of $1/n$ and the total probability of a query point falling into triangles in that group.  The threshold of $1/n$ prevents the group weight from being too small, allowing $T_p$ to be updated in $O(\log n)$ time.  The query time to locate a group is $O\bigl(\log\frac{W}{w}\bigr)$, where $w$ is the weight of that group and $W$ is the total weight in $T_p$.  Suppose that $D_1$ returns a triangle $t \in \Delta_1$ incident to $p$.  We find the group containing $t$ which tells us the face of $\cal S$ that contains $t$.  If $p$ is a boundary vertex of $\cal S$, there are two edges in $\partial {\cal S}$ incident to $p$, so we can check in $O(1)$ time whether $t$ lies in the exterior face.  Otherwise, we search $T_p$ to find the group containing $t$ in $O\bigl(\log\frac{W}{w}\bigr) = O\bigl(\log \frac{1}{\Pr(t)}\bigr)$ time.  

Given a query point $q$, we first query $D_m$ with $q$.  If a non-dummy triangle is reported by $D_m$, we are done.  Otherwise, we query $D_{m-1}$ and so on.  


\begin{lemma}
	\label{lem:dsexpected}
	Let ${\cal D} = ((\Delta_1,D_1), \ldots, (\Delta_m,D_m))$ be the data structure maintained for $\cal S$. The expected query time of ${\cal D}$ is $O(\mathrm{OPT})$, where {\rm OPT} is the minimum expected time of the best point location decision tree for $\cal S$.
\end{lemma}
\begin{proof}
We use $\mathrm{E}[{\cal D}]$ to denote the expected time of ${\cal D}$ to return a non-dummy triangle $t$ that contains the query point.   Let $Q_i$ denote the event that $t$ is reported at level $i$.  Let $\mathrm{E}[D_i |Q_i]$ denote the expected query time of $D_i$ conditioned on $Q_i$.

Consider $\mathrm{E}[{\cal D}|Q_i]$.  It must be the case that for $j \in [i+1,m]$, the search in $D_j$ returns a dummy triangle.  Each such search takes $O(\log n_j)$ time. 
It is known that $\sum_{j=i+1}^m \log_2 n_j = O(\log\log n_i)$~\cite[Claim~8]{paper:cheng2017adaptive}.  Therefore,
\[	
\mathrm{E}[{\cal D}| Q_i] \leq \mathrm{E}[D_i | Q_i] + O(\log \log n_i).
\]
For $i > 1$, each non-dummy triangle $t$ in $\Delta_i$ has weight $\max\{\Pr(t), W^*_i/n_i\}$, where $W^*_i$ is the sum of $\Pr(t)$ over all non-dummy triangles $t$ in $\Delta_i$.  Recall that the total weight $W_i$ of all triangles in $\Delta_i$ is  $\Theta(W^*_i)$.  Conditioned on $Q_i$, the probability of a non-dummy triangle $t \in \Delta_i$ containing the query point is $\Pr(t)/W^*_i$.  
The structure $D_i$, being the distribution-sensitive structure of Iacono~\cite{paper:Iacono2004}, guarantees that 
\[
\mathrm{E}[D_i | Q_i] = O\left(\sum_{\text{non-dummy} \, t  \,\in \Delta_i}
\frac{\Pr(t)}{W^*_i}\log\frac{W_i}{\Pr(t)}\right) = O\left(\sum_{\text{non-dummy} \, t  \,\in \Delta_i}
\frac{\Pr(t)}{W^*_i}\log\frac{W^*_i}{\Pr(t)}\right).
\]
According to the information-theoretic lower bound~\cite{paper:Shannon2001}, the rightmost sum above is the minimum expected query time, conditioned on $Q_i$, for returning the non-dummy triangle in $\Delta_i$ that contains the query point.  For $i = 1$, our point location structure in Section~\ref{sec:dyntridatastructure} (Lemma~\ref{lem:dyntrianalysis}) guarantees that $\mathrm{E}[D_1 | Q_1]$ is the asymtoptically minimum expected time plus an extra $O(\log \log n)$ overhead.

Let's do a mental exercise to locate $t$ differently as follows.  Let $\Delta_S$ be the Steiner triangulation of $S$ that has the minimum entropy.\footnote{It suffices for $\Delta_S$ to be a Steiner triangulation of near-minimum entropy as in~\cite[Theorem~2]{paper:Collette2012}.}  Let $D_{\Delta_S}$ denote the linear decision tree that takes the minimum expected time to locate a point in $\Delta_S$.  Query $D_{\Delta_S}$ to identify the triangle $t' \in \Delta_S$ that contains the query point.  Notice that $t'$ lies inside the same bounded face of $S$ that contains $t$.  So $t'$ intersects $O(\log n_i)$ non-dummy triangles in $\Delta_i$.  As a result, the intersection between $t'$ and the non-dummy triangles in $\Delta_i$ consists of $O(\log n_i)$ shapes of $O(1)$ complexities, and we can do a planar point location in $O(\log\log n_i)$ time to find the shape that contains the query point.  This shape is a part of $t$ which means that we have found $t$. The total expected time needed is $\mathrm{E}[D_{\Delta_S} | Q_i] + O(\log\log n_i)$.  Since $\mathrm{E}[D_i | Q_i]$ is the asymptotically minimum expected time to find $t$ in $\Delta_i$ for $i > 1$ and $\mathrm{E}[D_1 | Q_1]$ has an additive $O(\log \log n)$ overhead, we get $\mathrm{E}[D_i|Q_i] \leq O\left(\mathrm{E}[D_{\Delta_S}|Q_i] + \log\log n_i\right)$, which implies that
\[	
\mathrm{E}[{\cal D}| Q_i] \leq \mathrm{E}[D_i | Q_i] + O(\log \log n_i)
\leq O\left(\mathrm{E}[D_{\Delta_S} | Q_i] + \log\log n_i\right).
\]

We simplify the above inequality.  If $i = m$, then $O(\log\log n_m) = O(1)$ because $n_m = O(1)$.  Assume that $i \in [1,m-1]$.  Consider the triangles of $\Delta_S$ that are represented by leaves in $D_{\Delta_S}$ at depth $\log_2 \log_2 n_i$ or less.  Since $D_{\Delta_S}$ is a binary tree, there are $O(\log n_i)$ such leaves in $D_{\Delta_S}$.   Let $Z$ be the subset of non-dummy triangles in $\Delta_i$ that overlap with  such triangles of $\Delta_S$.  Each triangle in $\Delta_S$ interescts $O(\log n_i)$ non-dummy triangles in $\Delta_i$.  It follows that $\vert Z \vert = O(\log^2 n_i)$.  

Recall that level $i+1$ is constructed based on the $(\log_2 n_i)^4$ non-dummy triangles with the highest probabilities in $\Delta_i$.  The query point does not lie inside any of these $(\log_2 n_i)^4$ triangles; otherwise, the query point would be successfully located by $D_j$ for some $j > i$, contradicting the occurrence of the event $Q_i$.  Let $\Delta'_i$ be the subset of triangles in $\Delta_i$ that are not selected for the construction of level $i+1$.  
The query time of $D_{\Delta_S}$ to locate a point in any triangle in $\Delta'_i \setminus Z$ is $\Omega(\log\log n_i)$ by the definition of $Z$.  As the probabilities of triangles in $\Delta_i'$ containing a query point are not among the top $(\log_2 n_i)^4$ in $\Delta_i$, the probability of a triangle in $\Delta_i'$ containing a query point is at most $1/\log_2^4 n_i$.  Therefore, conditioned on $Q_i$, the probability of a query point falling in some triangle in $\Delta'_i \setminus Z$ is at least $1 - |Z|/\log^4 n_i
\geq 1 - O(1/\log^2 n_i)$.  As a result,
\[
\mathrm{E}[D_{\Delta_S} | Q_i] = \Omega\left(\left(1 - \frac{1}{\log^2n_i}\right)\log\log n_i\right)
= \Omega(\log\log n_i).
\]
Therefore, $\mathrm{E}[{\cal D} | Q_i] = O(\mathrm{E}[D_{\Delta_S} | Q_i])$.
Every query must be answered at exactly one level, i.e., $\sum_{i=1}^m \Pr(Q_i) = 1$. Thus,
\[
\mathrm{E}[{\cal D}] = \sum_{i=1}^m \Pr(Q_i) \cdot E[{\cal D}| Q_i]
= O\left(\sum_{i=1}^m \Pr(Q_i) \cdot E[D_{\Delta_S} | Q_i] \right)
= O\left(\mathrm{E}[D_{\Delta_S}]\right).
\]
It is known that $\mathrm{E}[D_{\Delta_S}] = O(\mathrm{OPT})$, where {\rm OPT} denotes the minimum expected time of the best point location decision tree for $S$~\cite{paper:Collette2012}.  We conclude that $\mathrm{E}[{\cal D}] = O(\mathrm{OPT})$.

If $Q_1$ happens, we have some extra work to do---finding the face of $S$ that contains the triangle $t$ of $\Delta_1$ in which the query pont is located.  This requires a search of some biased search tree $T_p$ for some vertex $p$ of $\Delta_1$ in $O\bigl(\log \frac{1}{\Pr(t)}\bigr)$ time.  Note that $\Omega\bigl(\log \frac{1}{\Pr(t)}\bigr)$ is an information-theoretic lower bound for $\mathrm{E}[D_1 | Q_1]$.  Therefore, this extra $O\bigl(\log \frac{1}{\Pr(t)}\bigr)$ search time can be absorbed by $\mathrm{E}[D_1 | Q_1]$ in the above analysis.
\end{proof}

\subsection{General-update sequence}
\label{subsec:convplsupdate}

Let $\Phi$ be a general-update sequence with $k \leq n/2$ edge updates.   We call $k$ the \emph{size} of $\Phi$.  As discussed in Section~\ref{sec:dk}, each edge update in $\cal S$ is transformed into $O(\log n)$ edge updates in $\Delta_1$.  Updating $\Delta_1$ takes $O(k\log^2 n)$ time.  We also update the biased search tree $T_p$ at each vertex $p$ of $\cal S$ affected by the structural changes in $\Delta_1$.  This step also takes $O(k\log^2 n)$ time.

For $i \geq 2$, we recompute $\Delta_i$ from $\Delta_{i-1}$ and then $D_i$ from $\Delta_i$.  By keeping the triangles of $\Delta_1$ in a max-heap according to the triangle probabilities, which can be updated in $O(k\log^2 n)$ time, we can extract the $n_2 = \log_2^4 n$ triangles to form $\Delta_2$ in $O(n_2\log n_2)$ time.   For $i \geq 3$, we scan $\Delta_{i-1}$ to extract the $n_i = \log_2^4 n_{i-1}$ triangles to form $\Delta_i$ in $O(n_{i-1} + n_i\log n_i)$ time.   For $i \geq 2$, constructing $D_i$ takes $O(n_i)$ time~\cite{paper:Iacono2004}.  The total update time of $\Delta_i$ and $D_i$ for $i \geq 2$ is $O\left(\sum_{i \geq 2} \log^4 n_{i-1}\log\log n_{i-1}\right)$, which telescopes to $O(\log^4 n \log\log n)$.

Consider $D_1$.  The second part of $D_1$ is a dynamic point location structure that admits an edge insertion/deletion in $\Delta_1$ in $O(\log^2 n)$ time, giving  $O(k\log^3 n)$ total time.
By Theorem~\ref{thm:3} in Section~\ref{sec:general}, the update time of the first part of $D_1$ is $O(k\log^5 n)$ amortized.

In the biased search tree $T_p$'s at the vertices $p$ of $\cal S$, there are different weight thresholds of $1/n$ depending on when a threshold was computed. To keep these thresholds within a constant factor of each other,  we rebuild the entire data structure periodically.  Let $n'$ be the number of vertices in the last rebuild.  Let $c < 1/2$ be a constant.  We rebuild when the total number of edge updates in $\cal S$ in all general-update sequences exceeds $cn'$ since the last rebuild.  Due to the first part of $D_1$, rebuidling takes $O(n\log^2 n)$ time by Theorem~\ref{thm:3}.  The second part of $D_1$ can also be constructed in $O(n\log^2 n)$ time.  This results in an extra $O(\log^2 n)$ amortized time per edge update in $\cal S$.

\begin{theorem}
	\label{thm:main}
	Suppose that there is a fixed but unknown query point distribution in $\real^2$, and there is an oracle that returns in $O(1)$ time the probability of a query point falling into a polygonal region of constant complexity.  There exists a dynamic point location structure for maintaining a convex subdivision $\cal S$ of $n$ vertices with the following guarantees.
	\begin{itemize}
		\item Any query can be answered in $O(\mathrm{OPT})$ expected time, where $\mathrm{OPT}$ is the minimum expected query time of the best point location linear decision tree for $\cal S$.
		\item The data structure uses $O(n\log^2 n)$ space, and it can be constructed in $O(n\log^2 n)$ time.
		\item A general-update sequence with size $k \leq n/2$ takes $O(k\log^5 n)$ amortized time.
	\end{itemize}
\end{theorem}


\begin{corollary}
\label{cor:main}
Given the same setting in Theorem~\ref{thm:main}, we can carry out a randomized incremental construction of the Voronoi diagram of $n$ sites in $O(n\log^5 n)$ expected time such that, for every $i \in [1,n]$, any nearest neighbor query after the insertions of the first $i$ sites can be answered in $O(\mathrm{OPT})$ expected time, where {\rm OPT} is the minimum expected query time of the best point location decision tree for the Voronoi diagram of the first $i$ sites.  The expectation of the Voronoi diagram construction time is taken over a uniform distribution of the permutations of the $n$ sites, whereas the expectation of the query time is taken over the query distribution.
\end{corollary}

\section{Slab tree: fixed vertical lines}
\label{sec:dyntridatastructure}

In this section, we present a static data structure for distribution-sensitive point location in a triangulation.  Its dynamization will be discussed in Sections~\ref{sec:update} and~\ref{sec:general}.

For any region $R \subset \real^2$, let $\Pr(R)$ denote the probability of a query point falling into $R$.   Let $\Delta$ be a triangulation with a convex outer boundary.  The vertices of $\Delta$ lie on a given set $\cal L$ of vertical lines, but some line in $\cal L$ may not pass through any vertex of $\Delta$.
For simplicity, we assume that no two vertices of $\Delta$ lie on the same vertical  line at any time.  

\begin{figure}[h]
	\centering
	\includegraphics[scale=0.55]{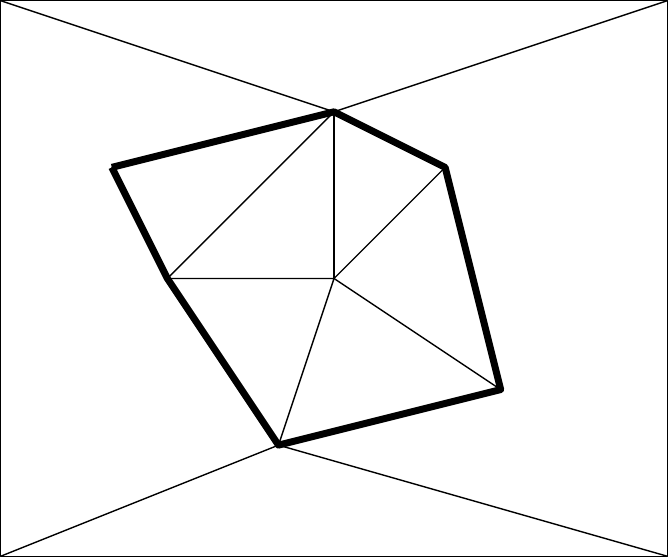}
	\caption{
		$\Delta_B$: the bold convex polygon bounds $\Delta$; the enclosing rectange is $B$; the two simple polygons in $B \setminus \Delta$ are triangulated by the algorithm of Hershberger and Suri.
	} 
	\label{fg:gaps_of_root}
\end{figure}

Enclose $\Delta$ with an axis-aligned bounding box $B$ such that no vertex of $\Delta$ lies on the boundary of $B$.   We assume that the left and right sides of $B$ lie on the leftmost and rightmost lines in $\cal L$.  Connect the highest vertex of $\Delta$ to the upper left and upper right corners of $B$, and then connect the lowest vertex of $\Delta$ to the lower left and lower right corners of $B$.  This splits $B \setminus \Delta$ into two triangles and two simple polygons.   Figure~\ref{fg:gaps_of_root} gives an example.  
The two simple polygons are triangulated using the method of Hershberger and Suri~\cite{paper:Hershberger1993}.   Let $\Delta_B$ denote the triangle tiling of $B$ formed by $\Delta$ and the triangulation of $B \setminus \Delta$.  Let $n$ denote the number of triangles in $\Delta_B$.  Any line segment in $B \setminus \Delta$ intersects $O(\log n)$ triangles in $\Delta_B$~\cite{paper:Hershberger1993}.  When we discuss updates in $\Delta$ later, the portion $\Delta_B \setminus \Delta$ of the tiling will not change although new vertices may be inserted into the outer boundary of $\Delta$.

Let $P$ denote the convex polygon bounded by the outer boundary of $\Delta$ and let ${\cal S}_P$ denote the subdivision with a single region $P$. We obtain a point location data structure $D_P$ of ${\cal S}_P$ by the method of Collette~et~al.~\cite{paper:Collette2012}. We use $D_P$ to determine if a query point falls outside $\Delta$.

\subsection{Structure definition}
\label{sec:struct}


Let $(l_1,l_2,\cdots,l_{|{\cal L}|})$ be the vertical lines in $\cal L$ in left-to-right order.  
We build the \emph{slab tree} $\cal T$ as follows.
The root of $\cal T$ represents the slab bounded by $l_1$ and $l_{|{\cal L}|}$.   The rest of $\cal T$ is recursively defined by constructing at most three children for every node $v$ of ${\cal T}$.  

We use $\slab(v)$ to denote the slab represented by $v$.  Let $(l_i,\cdots,l_k)$ be the subsequence of lines that intersect $\slab(v)$.   Choose $j \in [i,k)$ such that both the probabilties of a query point falling between $l_i$ and $l_j$ and between $l_{j+1}$ and $l_k$ are at most $\Pr(\slab(v))/2$.  Create the nodes $v_L$, $v_M$, and $v_R$ as the left, middle, and right children of $v$, respectively, where $\slab(v_L)$ is bounded by $l_i$ and $l_j$, $\slab(v_M)$ is bounded by $l_j$ and $l_{j+1}$, and $\slab(v_R)$ is bounded by $l_{j+1}$ and $l_k$.    No vertex of $\Delta_B$ lies in the interior of $v_M$.

The recursive expansion of $\cal T$ bottoms out at a node $v$ if $v$ is at depth $\log_2 n$ or $\slab(v)$ contains no vertex of $\Delta_B$ in its interior.   So the middle child of a  node is always a leaf.


Every node $v$ of $\cal T$ stores several secondary structures.  A connected region $R \subset \real^2$ \emph{spans} $v$ if there is a path $\rho \subset R \cap \slab(v)$ that intersects both bounding lines of $\slab(v)$.  The triangulation $\Delta_B$ induces a partition of  $\slab(v)$ into three types of regions:

\begin{itemize}
	
	\item {\sc Free Gap}:~For all triangle $t$ that spans $v$ but not $\mathit{parent}(v)$, $t \cap \slab(v)$ is a \emph{free gap} of $v$. 
	
	\item {\sc Blocked Gap}:~Let $E$ be the subset of all edges and triangles in $\Delta_B$ that intersect $\slab(v)$ but do not span $v$.  Every connected component in the intersection between $\slab(v)$ and the union of edges and triangles in $E$ is a \emph{blocked gap} of $v$.
	
	\item {\sc Shadow Gap}:~Take the union of the free gaps of all proper ancestors of $v$.  Each connected component in the intersection between this union and $\slab(v)$ is a \emph{shadow gap} of $v$.

\end{itemize}

\begin{figure}
	\centering
	\begin{tabular}{ccccccc}
		\includegraphics[scale=0.55]{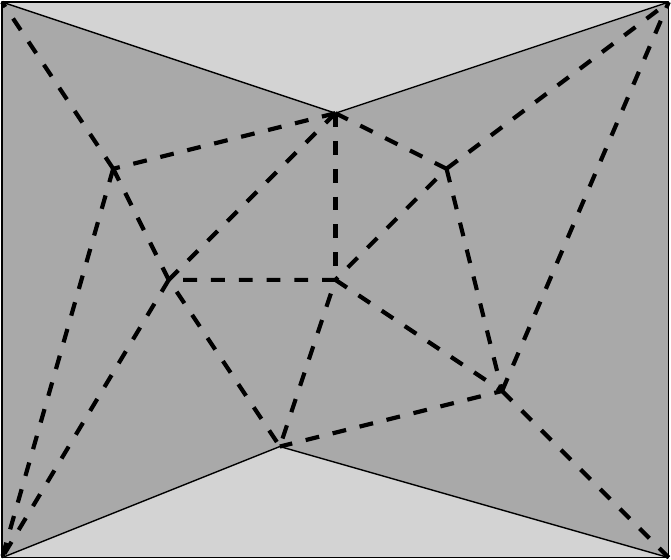} & &
		\includegraphics[scale=0.8]{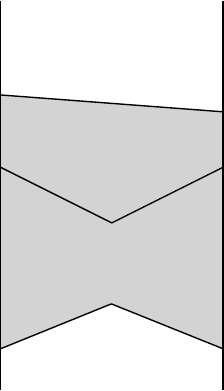} & &
		\includegraphics[scale=0.8]{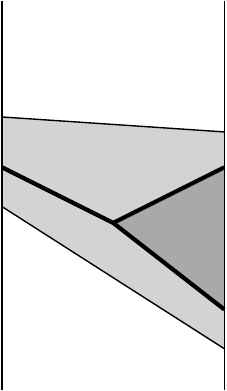} \\
		(a) & \hspace*{.2in} & (b) & \hspace*{.2in} & (c) 
	\end{tabular}
	\caption{(a)~Gaps of the root of $\cal T$: each light gray region is a free gap; the dark gray region is a blocked gap; the dashed segments show the triangles in $\Delta_B$ within the dark gray region. (b)~The two shaded gaps are adjacent. (c)~The two gaps in light gray are not adjacent; the gap consisting of the dark gray region and the bold edges is adjacent to the two gaps in light gray.}
	\label{fg:adjacent_gap}
\end{figure}

The upper boundary of a blocked gap $g$ has at most two edges, and so does the lower boundary of $g$.  If not, there would be a triangle $t$ outside $g$ that touches $g$, intersects $\slab(v)$, and does not span $v$.  But then $t$ should have been included in $g$, a contradiction.\footnote{There is one exception: when a blocked gap boundary contains a boundary edge $e$ of $\Delta$, updates may insert new vertices in the interior of $e$, splitting $e$ into collinear boundary edges.  However, the portion $\Delta_B \setminus \Delta$ of the triangle tiling remains fixed.  We ignore this exception to simplify the presentation.}

Two gaps of $v$ are \emph{adjacent} if the lower boundary of one is the other's upper boundary.  Figures~\ref{fg:adjacent_gap}(a)--(c) show some examples.

The list of free and blocked gaps of $v$ are stored in  vertical order in a balanced search tree, denoted by $\gaplist(v)$.  Group the gaps in $\gaplist(v)$ into maximal contiguous subsequences.  Store each such subsequence in a biased search tree~\cite{paper:bent85} which allows an item with weight $w$ to be accessed in $O\bigl(\log\frac{W}{w}\bigr)$ time, where $W$ is the total weight of all items.  The weight of a gap $g$ set to be $\Pr(g)$.    We call each such biased search tree a \emph{gap tree} of $v$.

For every internal node $v$ of $\cal T$, we set up some pointers from the gaps of $v$ to  the gap trees of the children of $v$ as follows.  Let $w$ be a child of $v$.  The free gaps of $v$ only give rise to shadow gaps of $w$, so they do not induce any item in $\gaplist(w)$.   Every blocked gap $g$ of $v$ gives rise to a contiguous sequence $\sigma$ of free and blocked gaps of $w$.   Moreover, $\sigma$ is maximal in $\gaplist(w)$ because $g$ is not adjacent to any other blocked gap of $v$.  Therefore, $\sigma$ is stored as one gap tree $T_\sigma$ of $w$.  We keep a pointer from $g$ to the root of $T_\sigma$.  


Since we truncate the recursive expansion of the slab tree $\cal T$ at depth $\log_2 n$, we may not be able to answer every query using $\cal T$.  We need a backup which is a dynamic point location structure ${\cal T}^*$~\cite{paper:Arge2006, paper:Chan2015, paper:Chiang1992, paper:Preparata1989}.  Any worst-case dynamic point location structure with $O(\log n)$ query time suffices, provided that its update time is $O(\log^2 n)$ and its space is $O(n\log n)$.  

\subsection{Querying}
\label{subsec:loglognanalysis}

Given a query point $q$, we determine if $q$ lies inside $\Delta$ by $D_P$. If $q$ lies outside $\Delta$, we output that $q$ lies outside ${\cal S}$. If $q$ lies inside $\Delta$, we start at the root $r$ of $\cal T$, and $q$ must lie in a gap stored in the only gap tree of $r$.  In general, when we visit a node $v$ of $\cal T$, we also know a gap tree $T_v$ of $v$ such that $q$ lies in one of the gaps in $T_v$.  We search $T_v$ to locate the gap, say $g$, that contains $q$.  If $g$ is a free gap, the search terminates because we have located a triangle in $\Delta_B$ that contains $q$.  Suppose that $g$ is a blocked gap.  Then, we check in $O(1)$ time which child $w$ of $v$ satisfies $q \in \slab(w)$.  By construction, $g$ contains a pointer to the gap tree $T_w$ of $w$ that stores the free and blocked gaps of $w$ in $g \cap \slab(w)$.  We jump to $T_w$ to continue the search.  If the search reaches a leaf of $\cal T$ without locating a triangle of $\Delta_B$, we answer the query using ${\cal T}^*$.

We need the following technical result to analyze the expected query time.

\begin{lemma}
	\label{lem:loglognoverhead}
	Let {\rm OPT} be the expected query time of the best point location decision tree for $\Delta$.  Let $H(\Delta_B)$ be the entropy of $\Delta_B$.   Then, $H(\Delta_B) = O(\mathrm{OPT} + \log \log n)$.
\end{lemma}
\begin{proof}
	Any query point that falls outside $B$ can be detected in $O(1)$ time.   Consider the case of a query point falling inside $B$ and hence inside some triangle in $\Delta_B$.   By the result of Collette~et~al.~\cite{paper:Collette2012}, there is a linear decision tree $D^*$  for anwering queries in $\Delta$ with expected query time $O(\text{OPT})$ such that every leaf of $D^*$ represents a triangle that lies inside a triangle of $\Delta$ or the exterior region of $\Delta$.   
	
	Since $\Delta_B \setminus \Delta$ is triangulated using the method of Hershberger and Suri~\cite{paper:Hershberger1993}, every leaf triangle of $D^*$ that lies in the exterior region of $\Delta$ intersects $O(\log n)$ triangles of $\Delta_B$.  Every triangle of $\Delta$ is also a triangle of $\Delta_B$.  Thus, each leaf triangle of $D^*$ intersects $O(\log n)$ triangles of $\Delta_B$.   
	
	Suppose that the query point $q$ is located in a leaf triangle $t$ of $D^*$.  The triangles in $\Delta_B$ that intersect $t$ induce a planar subdivision in $t$ of size $O(\log n)$.  Using a static, worst-case optimal planar point location structure~(e.g.,~\cite{paper:Kirkpatrick1981,paper:Sarnak1986}), we can thus determine the triangle of $\Delta_B$ containing $q$ in an extra $O(\log \log n)$ time.  Therefore, $D^*$ can be extended to answer queries in $\Delta_B$ in $O(\text{OPT} + \log \log n)$ expected time.  By Shannon's work~\cite{paper:Shannon2001}, $H(\Delta_B)$ is the lower bound for answering queries in $\Delta_B$ in the comparison-based model. As a result, $H(\Delta_B) = O(\text{OPT} + \log \log n)$.
\end{proof}

The analysis of the expected query time of the slab tree $\cal T$ exploits two facts: the halving of the probabilities of a query point falling into the slabs of internal nodes along a root-to-leaf path in $\cal T$, and storing gap trees as biased search trees.  They let us zoom into the target quickly.  The $\log\log n$ term arises because a triangle $t$ induces $O(\log n)$ free gaps, thus adding $\Pr(t)\log\log n$ to $t$'s contribution to the entropy.

\begin{lemma}
	\label{lem:dyntrianalysis}
	The expected query time of $\cal T$ is $O(\mathrm{OPT} + \log \log n)$, where OPT is the expected query time of the best point location decision tree for $\Delta$.
\end{lemma}
\begin{proof}
The data structure $D_P$ is constructed from the method proposed by Collette~et~al.~\cite{paper:Collette2012}. This data structure provides the asymptotically minimum expected query time for deciding whether a query point is inside ${\cal S}$ or outside ${\cal S}$. Any point location structure for $\Delta$ must make the same decision. So the expected query time of $D_P$ is $O(\text{OPT})$.

Let $q$ be a query point. If $q$ lies outside $\Delta$, then $q$ is answered in $O(\text{OPT})$.
Suppose $q$ lies inside $\Delta$. Let $\tau$ be the triangle in $\Delta_B$ that contains $q$.  Let $w$ be the node of $\cal T$ at which the search terminates.   When searching in ${\cal T}$, we alternate between locating $q$ in a finer slab and locating $q$ in a finer gap.   We first analyze the total time spent on visiting finer slabs. 

The root of ${\cal T}$ is at depth 0.  Every internal node $v$ of $\cal T$ has at most three children $v_L$, $v_M$, and $v_R$.  The probabilities $\Pr(\slab(v_L))$ and $\Pr(\slab(v_R))$ are at most $\Pr(\slab(v))/2$, and $v_M$ is a leaf.   It follows that for each node $w$ of $\cal T$, $\Pr(\slab(w)) \leq 2^{1-\text{depth}(w)}$, which implies that 
$\text{depth}(w) \leq 1 - \log_2 \Pr(\slab(w))$.  The total time spent on locating finer slabs is $O(\text{depth}(w)) = O(- \log\Pr(\slab(w)))$.

The total time  spent on locating finer gaps is the total query time of the gap trees.    For $i \in [0,\text{depth}(w)]$, let $T_i$ denote the gap tree that we visited at depth $i$ in $\cal T$ during the search, and let $g_i$ denote the gap in $T_i$ that contains $q$.  The weight of $g_i$ in $T_i$ is at least the total weight of $T_{i+1}$ because all free and blocked gaps in $T_{i+1}$ are subsets of $g_i$.    Note that $\slab(w) \cap \tau \subseteq g_{\text{depth}(w)}$ and $\text{weight}(T_0) \leq 1$.  The total query time of the gap trees is
\begin{eqnarray*}
	O \left (\sum_{i = 0}^{\mathrm{depth}(w)} \left (\log \frac{\text{weight}(T_i)}{\text{weight}(g_i)} + 1 \right ) \right ) 
	& = & O \left ( \log \frac{\text{weight}(T_0)}{\Pr(\slab(w) \cap \tau)} + \mathrm{depth}(w) \right) \\
	&=& O\left (\log \frac{1}{\Pr(\slab(w) \cap \tau)}\right ).
\end{eqnarray*}

In summary, the total search time is $O\bigl (-\log \Pr(\slab(w) \cap \tau) \bigr)$.  Hence, the contribution of $\slab(w) \cap \tau$ to the expected query time of $\cal T$ is $O\bigl (-\Pr(\slab(w) \cap \tau) \log \Pr(\slab(w) \cap \tau) \bigr)$.  Either $\slab(w) \cap \tau$ is a free gap of $w$, or $\tau$ does not span $w$ but the search terminates as $w$ is at depth $\log_2 n$.  For every triangle $t \in \Delta_B$,  define 
${\cal F}_t$ to be the collection of $\slab(v) \cap t$ over all slab tree nodes $v$ such that either $\slab(v) \cap t$ is a free gap of $v$, or $v$ is a leaf node and $\slab(v) \cap t$ is contained in a blocked gap of $v$.  Note that ${\cal F}_t$ is a partition of $t$.   By our previous conclusion, the expected time of querying $\cal T$ is 
\[
O\left (\sum_{t \in \Delta_B} \sum_{f \in {\cal F}_t} \Pr(f) \log \frac{1}{\Pr(f)} \right ).
\]
We bound this quantity as follows.  Akin to storing intervals in a segment tree, we have $|{\cal F}_t| = O(\log n)$.   For each region $f \in {\cal F}_t$,  define $\alpha_f = \Pr(f)/\Pr(t)$.  Then,
\begin{eqnarray*}
	\sum_{f \in {\cal F}_t} \Pr(f) \log \frac{1}{\Pr(f)} &=& \sum_{f \in {\cal F}_t} \alpha_f \Pr(t) \log \frac{1}{\alpha_f\Pr(t)}\\
	&=& \Pr(t)\log \frac{1}{\Pr(t)}\sum_{f \in {\cal F}_t} \alpha_f + \Pr(t) \sum_{f \in {\cal F}_t} \alpha_f \log \frac{1}{\alpha_f}.
\end{eqnarray*}
Note that $\sum_{f \in {\cal F}_t} \alpha_f = 1$.   Also, $\sum_{f \in {\cal F}_t} \alpha_f \log (1/\alpha_f)$ is maximized when $\alpha_f = 1/|{\cal F}_t|$ for all $f \in {\cal F}_t$.  Therefore,
\begin{eqnarray*}
	\sum_{f \in {\cal F}_t} \Pr(f) \log \frac{1}{\Pr(f)} & \leq & \Pr(t) \log \frac{1}{\Pr(t)} + O(\Pr(t) \log \log n).
\end{eqnarray*}
Hence,
\begin{eqnarray*}
	O\left (\sum_{t \in \Delta_B} \sum_{f \in {\cal F}_t} \Pr(f) \log \frac{1}{\Pr(f)} \right )
	&=& O\left ( \sum_{t \in \Delta_B} \Pr(t) \log \frac{1}{\Pr(t)} + \log \log n \sum_{t \in \Delta_B} \Pr(t) \right ) \\
	&=& O(H(\Delta_B) + \log \log n) \\
	& = & O(\text{OPT} + \log\log n).  \quad\quad\quad (\because\text{Lemma~\ref{lem:loglognoverhead}})
\end{eqnarray*}
So far, we have ignored the event of querying the backup point location structure.  This happens when $\text{depth}(w) = \log_2 n$ and $\slab(w) \cap \tau$ is not a free gap.  Querying the backup structure takes $O(\log n)$ time which is $O(\text{depth}(w))$ in this case.  Thus, there is no asymptotic increase in the expected query time.
\end{proof}

\subsection{Construction}
\label{sec:construct}

The children of a node $v$ of $\cal T$ can be created in time linear in the number of lines in $\cal L$ that intersect $\slab(v)$.  Thus, constructing the primary tree of $\cal T$ takes $O(|{\cal L}|\log n)$ time.  

The gap lists and gap trees are constructed via a recursive traversal of $\cal T$.   In general, when we come to a node $v$ of $\cal T$ from $\mathit{parent}(v)$, we maintain the following preconditions.

\begin{itemize}
	\item We have only those triangles in $\Delta_B$ such that each intersects $\slab(v)$ and does not span $\mathit{parent}(v)$.   These triangles form a directed acyclic graph $G_v$:~triangles are graph vertices, and two triangles sharing a side are connected by a graph edge directed from the triangle above to the one below.\footnote{Refer to~\cite[Section~4]{paper:Edelsbrunner1986} for a proof that this ordering is acyclic.}
	
	\item 
	The connected components of $G_v$ are sorted in order from top to bottom.  Note that each connected component intersects both bounding lines of $\slab(v)$.
\end{itemize}


Each connected component $C$ in $G_v$ corresponds to a maximum contiguous subsequence of free and blocked gaps in $\gaplist(v)$ (to be computed), so for each $C$, we will construct a gap tree $T_C$.  We will return the roots of all such $T_C$'s to $\mathit{parent}(v)$ in order to set up pointers from the blocked gaps of $\mathit{parent}(v)$ to the corresponding $T_C$'s.
	
\paragraph*{Gap list}  We construct $\gaplist(v)$ first.  Process the connected components of $G_v$ in vertical order.  Let $C$ be the next one.  The restriction of the upper  boundary of $C$ to $\slab(v)$ is the upper gap boundary induced by $C$.  
Perform a topological sort of the triangles in $C$.  We pause whenever we visit a triangle $t \in C$ that spans $v$.  Let $t'$ denote the last triangle in $C$ encountered that spans $v$, or in the absence of such a triangle, the upper boundary of $C$.  If $t \cap t' = \emptyset$ or $t \cap t'$ does not span $v$, the region in $\slab(v)$ between $t'$ and $t$ is a blocked gap, and we append it to $\gaplist(v)$.   Then, we append $\slab(v) \cap t$ as a newly discovered free gap to $\gaplist(v)$.  The construction of $\gaplist(v)$ takes $O(|G_v|)$ time.

\paragraph*{Recurse at the children}  
Let $v_L$, $v_M$ and $v_R$ denote the left, middle and right children of $v$.   We scan the connected components of $G_v$ in the vertical order to extract $G_{v_L}$.  A connected component $C$ in $G_v$ may yield multiple components in $G_{v_L}$ because the triangles that span $v$ are omitted.  The components in $G_{v_L}$ are ordered vertically by a topological sort of $C$.
Thus, $G_{v_L}$ and the vertical ordering of its connected components are produced in $O(|G_v|)$ time.   
The generation of $G_{v_M}$, $G_{v_R}$ and the vertical orderings of their connected components is similar.  Then, we recurse at $v_L$, $v_M$ and $v_R$.

\paragraph*{Gap trees}  After we have recursively handled the children of $v$, we construct a gap tree for each maximal contiguous subsequence of gaps in $\gaplist(v)$.  The construction takes linear time~\cite{paper:bent85}.   The recursive call at $v_L$ returns a list, say $X$, of the roots of gap trees at $v_L$, and $X$ is sorted in vertical order.   There is a one-to-one correspondence between $X$ and the blocked gaps of $v$ in vertical order.  Therefore, in $O(|\gaplist(v)|)$ time, we can set up pointers from the blocked gaps of $v$ to the corresponding gap tree roots in $X$.  The pointers from the blocked gaps of $v$ to the gap tree roots at $v_M$ and $v_R$ are set up in the same manner.  Afterwards, if $v$ is not the root of $\cal T$, we return the list of gap tree roots at $v$ in vertical order.  

\paragraph*{Running time}  We spend $O(|G_v|)$ time at each node $v$.   If a triangle $t$ contributes to $G_v$ for some node $v$, then either $\slab(v) \cap t$ is a free gap of $v$, or $\slab(v) \cap t$ is incident to the leftmost or rightmost vertex of $t$.    
Like storing segments in a segment tree, $t$ contributes $O(\log n)$ free gaps.
The nodes of $\cal T$ whose slabs contain the leftmost (resp.~rightmost) vertex of $t$ form a root-to-leaf path.  Therefore, $t$ contributes $O(\log n)$ triangles in the $G_v$'s over all nodes $v$ in $\cal T$.  The sum of $|G_v|$ over all nodes $v$ of $\cal T$ is $O(n\log n)$.

\begin{lemma}
	\label{lem:slabtree}
	Given $\Delta_B$ and $\cal L$, the slab tree and its auxiliary structures, including gap lists and gap trees, can be constructed in $O(|{\cal L}|\log n)$ time and $O(n\log n)$ space.
\end{lemma}

\section{Handling triangulation-updates: fixed vertical lines}
\label{sec:update}

We discuss how to update the slab tree when $\Delta_B$ is updated such that every new vertex lies on a vertical line in the given set $\cal L$.  This restriction will be removed later in Section~\ref{sec:general}.  A \emph{triangulation-update} $U$ has the following features:
\begin{itemize}
	\item It specifies some triangles in $\Delta$ whose union is a polygon $R_U$ possibly with holes.
	
	\item It specifies a new triangulation $T_U$ of $R_U$.   $T_U$ may contain vertices in the interior of $R_U$.  $T_U$ does not have any new vertex in the boundary of $R_U$, except possibly for the boundary edges of $R_U$ that lie on the outer boundary of $\Delta$.
	
	\item The construction of $T_U$ takes $O(|T_U| \log |T_U|)$ time.
	
	\item The \emph{size} of $U$ is the total number of triangles in $\Delta \cap R_U$ and $T_U$.
\end{itemize}


Our update algorithm is a localized version of the construction algorithm in Section~\ref{sec:construct}.  It is also based on a recursive traversal of the slab tree $\cal T$.
%
%
%
When we visit a node $v$ of $\cal T$, we have a directed acyclic graph $H_v$ that represents \emph{legal} and \emph{illegal} regions in $T_U \cap \slab(v)$: 
\begin{itemize}
	\item For each triangle $t \in T_U$ that intersects the interior of $\slab(v)$ and does not span $\mathit{parent}(v)$, $t \cap \slab(v)$ is a legal region in $H_v$.
	\item Take the triangles in $T_U$ that span $\mathit{parent}(v)$.  Intersect their union with $\slab(v)$.  Each resulting connected component that has a boundary vertex in the interior of $\slab(v)$ is an \emph{illegal region}.  Its upper and lower boundaries contain at most two edges each.
	Requiring a boundary vertex inside $\slab(v)$ keeps the complexity of illegal regions low.
	\item Store $H_v$ as a directed acyclic graph: regions are graph vertices, and two regions sharing a side are connected by an edge directed from the region above to the one below.   The directed acyclic graph $H_v$ may not be connected.  We use $\widetilde{H}_v$ to denote the subset of connected components in $H_v$ that intersect both bounding lines of $\slab(v)$.  The connected components in $\widetilde{H}_v$ are given in sorted order from top to bottom.  As we will see later, the ordering of the remaining connected components in $H_v \setminus \widetilde{H}_v$ is  unimportant with respect to the update at $v$.
\end{itemize}

An overview of the update procedure is as follows.  We update the auxiliary structures of the slab tree $\cal T$ in a recursive traversal of it.  Suppose that we visit a node $v$ of $\cal T$ in the traversal.   We update $\gaplist(v)$ and then recursively visit the children $v_L$, $v_M$ and $v_R$ of $v$.   The recursive calls return three lists \emph{updated-trees}$(v_L)$, \emph{updated-trees}$(v_M)$ and \emph{updated-trees}$(v_R)$ that store the roots of those gap trees at $v_L$, $v_M$ and $v_R$, respectively, that are affected by the triangulation-update.  We set pointers from the appropriate blocked gaps in $\gaplist(v)$ to the gap trees in \emph{updated-trees}$(v_L)$, \emph{updated-trees}$(v_M)$, and \emph{updated-trees}$(v_R)$.  Afterwards, we construct a list, \emph{updated-trees}$(v)$, of the roots of the gap trees of $v$ that are affected by the triangulation-update.  Finally, if $v$ is not the root of $\cal T$, we return the list \emph{updated-trees}$(v)$ to $\mathit{parent}(v)$.  If $v$ is the root of $\cal T$, $\slab(v)$ contains one free gap, one blocked gap and another free gap in this order, and there is no change to these three gaps no matter what triangulation-updates have happened.

We describe the details of the update procedure in Sections~\ref{sec:update-gaplist}--\ref{sec:rebuild}.  Given a connected region $R$ that lies inside and spans a slab, we use $\up(R)$ and $\low(R)$ to denote the upper and lower boundaries of $R$, respectively.

\subsection{Updating the gap list at a slab tree node}
\label{sec:update-gaplist}

\subsubsection{Preliminaries}

We first show that every component in $H_v \setminus \widetilde{H}_v$ is contained in a single blocked gap before and after the triangulation-update.  This justifies ignoring $H_v \setminus \widetilde{H}_v$ in our subsequent processing.

\begin{lemma}
	\label{lem:ignore}
	Every connected component of $H_v \setminus \widetilde{H}_v$ is part of a blocked gap before and after the triangulation-update.
\end{lemma}
\begin{proof}
	Let $C$ be a connected component of $H_v \setminus \widetilde{H}_v$.  Since $C$ intersects at most one bounding line of $\slab(v)$, every edge and triangle in $\Delta_B$ whose intersection with $\slab(v)$ belongs to $C$ cannot span $v$.  Therefore, $C$ is contained in a connected component of the intersection between $\slab(v)$ and the union of edges and triangles in $\Delta_B$ that intersect $\slab(v)$ but do not span $v$, i.e., a blocked gap.
\end{proof}

We show that it suffices to check $\gaplist(v)$ and $\widetilde{H}_v$ to update the gaps of $v$.

\begin{lemma}
	\label{lem:ignore2}
	Let $g$ be a free or blocked gap of $v$ after a triangulation-update.  For all $\gamma \in \{\up(g), \low(g)\}$, one of the following properties is satisfied:
	\begin{itemize}
		\item $\gamma$ is contained in some component in $\widetilde{H}_v$;
		\item $\gamma$ is the upper or lower boundary of a gap of $v$ before that triangulation-update.
	\end{itemize}
\end{lemma}
\begin{proof}
	We prove the lemma for $\up(g)$.  Similar reasoning applies to $\low(g)$.
	
	If $g$ is a free gap, then $g = t \cap \slab(v)$ for some triangle $t$ in the new triangulation that spans $v$ but not $\mathit{parent}(v)$.  If $t$ exists in the old triangulation, then $g = t \cap \slab(v)$ was a free gap of $v$ before the triangulation-update.  If $t$ is new, then $t$ must be a triangle in $T_U$.  Hence, $t \cap \slab(v)$ is contained in $\widetilde{H}_v$ because $t$ spans $v$ but not $\mathit{parent}(v)$.  It follows that $\up(g) = \up(t \cap \slab(v))$ is contained in $\widetilde{H}_v$.
	
	Suppose that $g$ is a blocked gap.  
	There is a free or shadow gap $g'$ of $v$ with respect to the new triangulation such that $\low(g') = \up(g)$.  
	Note that $\low(g') = \low(t' \cap \slab(v))$ for some triangle $t'$ in the new triangulation that spans $v$.   If $t'$ is a triangle in $T_U$, then $\up(g) = \low(g') = \low(t'\cap \slab(v))$ is contained in $\widetilde{H}_v$.  Suppose that $t'$ exists in the old triangulation.   Let $g''$ be the free or shadow gap of $v$ with respect to the old triangulation that contains $t' \cap slab(v)$.
	If $\low(g'') = \low(t'\cap \slab(v))$, we are done.  Otherwise, $g''$ is a shadow gap, and $g''$ contains another triangle $t''$ in the old triangulation such that $t''$ is below $t'$ and $t' \cap t''$ is an edge that spans $v$.  Note that $\up(g) = \low(g') = \up(t'' \cap \slab(v))$.  The triangle $t''$ ceases to exist after the triangulation-update because $t''$ overlaps with the blocked gap $g$. Therefore, $\up(g) = \up(t'' \cap \slab(v))$ must be contained in $\widetilde{H}_v$. 
\end{proof}

\subsubsection{Updating the gap list}
\label{sec:gaplist}

By Lemma~\ref{lem:ignore2}, it suffices to check $\gaplist(v)$ and $\widetilde{H}_v$ to update the gaps of $v$.  Let $C_1,C_2,\cdots,C_m$ denote the connected components in $\widetilde{H}_v$ in order from top to bottom.   Each $C_i$ has an upper boundary $\up(C_i)$ and a lower boundary $\low(C_i)$.  The rest of the  boundary of $C_i$ may include segments on the boundary of $\slab(v)$ and polygonal chains that have both endpoints on the same bounding line  of $\slab(v)$, but these boundary portions will not be relevant for our discussion.  We use $\reg(C_i)$ to denote the union of regions in $C_i$.

We process $C_1, C_2,\ldots, C_m$ in this order.  We maintain several variables whose definitions and initializations are explained below.
\begin{itemize}
	\item A balanced search tree $L_v$.   We initialize $L_v:= \gaplist(v)$ before processing $C_1$.  The breaks between maximal contiguous subsequences in $L_v$ are the shadow gaps of $v$, so boundaries of shadow gaps can be retrieved quickly.  We will update $L_v$ incrementally and set $\gaplist(v)$ to be the final $L_v$.
	
	\item A variable $\mathit{state} \in \{ \mathsf{blocked}, \mathsf{null} \}$.  If we are currently building a blocked gap, then $\mathit{state} = \mathsf{blocked}$; otherwise, $\mathit{state} = \mathsf{null}$.  We initialize $\mathit{state}:= \mathsf{null}$ before processing $C_1$.
	
	\item A variable $\gamma$ that keeps track of the upper boundary of the current blocked gap being built.  The content of $\gamma$ is only valid when $\mathit{state} = \mathsf{blocked}$.  We may update $\gamma$ even if $\mathit{state} = \mathsf{null}$ so that the content of $\gamma$ will be valid when $\mathit{state}$ becomes $\mathsf{blocked}$.
	
	\item A balanced search tree $\Sigma$ that keeps track of some free and blocked gaps being built to replace certain gaps in $L_v$.   Every now and then, certain gaps in $L_v$ will be replaced by the gaps in $\Sigma$.  Afterwards, $\Sigma$ will be emptied.  We initiliaze $\Sigma:= \emptyset$ before processing $C_1$.
	
	
\end{itemize}

The following procedure {\sc Modify}$(v)$ processes $C_1, C_2, \cdots, C_m$ to update $\gaplist(v)$.    The update is an incremental merge of these components with the old version of $\mathit{gaplist}(v)$.

\vspace{10pt}

\noindent {\sc Modify}$(v)$:
\begin{enumerate}
	
	\item $L_v:= \gaplist(v)$; $\mathit{state}:= \mathsf{null}$; $\Sigma:= \emptyset$; 
	
	\item $i:= 1$.
	
	\item Set $\hat{g}_i$ to be the gap described in criterion (i), (ii), or (iii) below, whichever is applicable.  If more than one criterion is applicable, the order of precedence is (i), (ii), (iii).
	\begin{itemize}
		\item[(i)] The blocked gap in $L_v$ whose interior or boundary intersects $\up(C_i)$.
		\item[(ii)] The free gap in $L_v$ intersected by $\up(C_i)$ and $\reg(C_i)$.
		\item[(iii)] The shadow gap with respect to $L_v$ intersected by $\up(C_i)$ and $\reg(C_i)$.
	\end{itemize}
	/* Note that $\up(\hat{g}_i)$ and $\up(C_i)$ do not cross */
	
	\item If $\mathit{state} = \mathsf{null}$ then \{ \\
	\hspace*{12pt}\parbox[t]{5in}{
		if $\hat{g}_i$ is a blocked gap and $\up(\hat{g}_i)$ is above or partly above $\up(C_i)$ then \{ \\
		\hspace*{20pt}$\mathit{state} := \mathsf{blocked}$; \\
		\hspace*{20pt}$\gamma := \up(\hat{g}_i)$ } \\
	\hspace*{18pt} \} \\
	\hspace*{12pt}else \\
	\hspace*{20pt}$\gamma := \up(C_i)$  \quad /* we may discover later that a new blocked gap begins at $\up(C_i)$ */ \\
	\hspace*{6pt}\}
	
	\item Perform a topological sort of $C_i$.  For each region $\tau \in C_i$ encountered, \\
	\hspace*{12pt}if $\tau$ is a legal region and $\tau$ does not span $v$ then \\
	\hspace*{24pt}$\mathit{state}:= \mathsf{blocked}$ \\
	\hspace*{12pt}else \{ \\
	\hspace*{24pt}\parbox[t]{5in}{
		if $\mathit{state} = \mathsf{blocked}$ then \{ \\
		\hspace*{12pt}~/* a new blocked gap ends at $\up(\tau)$ */ \\
		\hspace*{12pt}~make a blocked gap $g$ bounded between $\gamma$ and $\up(\tau)$; \\
		\hspace*{12pt}~append $g$ to $\Sigma$ \\
		\hspace*{8pt} \} \\
		if $\tau$ is a legal region, then append $\tau$ to $\Sigma$ as a new free gap; \\
		$\gamma := \low(\tau)$; \\
		$\mathit{state}:= \mathsf{null}$ } \\
	\hspace*{12pt}\}
	
	\item The topological sort of $C_i$ ends when we come to this step.  
	\begin{enumerate}
		
		\item Set $\check{g}_i$ to be the gap described in criterion (i), (ii), or (iii) below, whichever is applicable.  If more than one criterion is applibcable, the order of precedence is (i) (ii), (iii).
		\begin{itemize}
			\item[(i)] The blocked gap in $L_v$ whose interior or boundary intersects $\low(C_i)$.
			\item[(ii)] The free gap in $L_v$ intersected by $\low(C_i)$ and $\reg(C_i)$.
			\item[(iii)] The shadow gap with respect to $L_v$ intersected by $\low(C_i)$ and $\reg(C_i)$.
		\end{itemize}
		/* Note that $\low(\check{g}_i)$ and $\low(C_i)$ do not cross */
		
		\item If $\mathit{state} = \mathsf{blocked}$ and ($i = m$ or $\check{g}_i$ is disjoint from $\up(C_{i+1})$),then \{ \\
		\hspace*{20pt}\parbox[t]{5in}{
			make a blocked gap $g$ bounded between $\gamma$ and $\low(\check{g}_i)$; \\
			append $g$ to $\Sigma$; \\
			replace by $\Sigma$ the gaps in $L_v$ from $\hat{g}_i$ to $\check{g}_i$; /* $\hat{g}_i$ and $\check{g}_i$ are also replaced */ \\
			$\mathit{state} := \mathsf{null}$ }   \\
		\hspace*{12pt} \} \\
		else \{ \\
		\hspace*{20pt}\parbox[t]{5in}{ 
			if $\low(C_i)$ intersects the interior of $\check{g}_i$ then \\
			\hspace*{16pt}split $\check{g}_i$ in $L_v$ at $\low(C_i)$ into two gaps; \\
			replace by $\Sigma$ the gaps in $L_v$ from $\hat{g}_i$ to $\low(C_i)$ \\
			/* $\hat{g}_i$ and the gap immediately above $\low(C_i)$ are also replaced */ } \\
		\hspace*{12pt}\}
		
		\item $\Sigma:= \emptyset$.
		\item If $i < m$, then $i:= i+1$ and go to step~3.
	\end{enumerate}
	
	\item $\gaplist(v):= L_v$.
	
\end{enumerate}

\begin{lemma}
	\label{lem:modify-time}
	{\sc Modify}$(v)$ updates $\mathit{gaplist}(v)$ correctly in $O(|H_v|\log n)$ amortized time.
\end{lemma}
\begin{proof}
The correctness is established by induction on the processing of the components $C_1,C_2,\ldots,C_m$.  The proof of correctness is deferred to Appeneix~\ref{app:modify}.  The running time of {\sc Modify}$(v)$ is clearly $O(|H_v|\log n)$ plus the time to delete gaps in $L_v$ that are replaced in step~6(b) of {\sc Modify}$(v)$.  Each such deletion takes $O(\log n)$ time.  In Section~\ref{sec:rebuild}, we will introduce a periodic rebuild of $\cal T$ and its auxiliary structures so that $n = \Theta(n')$, where $n'$ is the number of triangles in $\Delta_B$ in the initial construction or the last rebuild, whichever is more recent.  The gap deleted from $L_v$ might be inserted in the past since the initial construction or the last rebuild, or in the initial construction or the last rebuild, whichever is applicable and more recent.  We charge the $O(\log n)$ deletion time to the insertion of that deleted gap.  Note that we might have spent as little as $O(\log n')$ time in inserting that gap into $L_v$ in the past.  Nevertheless,  $\log n = \Theta(\log n')$ as $n = \Theta(n')$, and therefore, the charging argument goes through.  Thus, the total running time is $O(|H_v|\log n)$ amortized.
\end{proof}


\subsection{Recurse at children and return from recursions at children}

After running {\sc Modify}$(v)$ at a slab tree node $v$, we recurse at the children of $v$.  Let $w$ denote any child of $v$.   Recursing at $w$ requires the construction of $H_w$ and $\widetilde{H}_w$ from $H_v$ and $\widetilde{H}_v$ which is described in the following.

We first construct a balanced search tree $S_w$ of legal and candidate illegal regions in $\widetilde{H}_w$.  All legal regions in $\widetilde{H}_w$ will be included as legal regions in $S_w$.  All illegal regions in $\widetilde{H}_w$ will be included as candidate illegal regions in $S_w$.  However, some of the candidate illegal regions in $S_w$ have no vertex in the interior of $\slab(w)$, so they will be removed later.  The construction of $S_w$ goes through two stages.

The first stage processes $\widetilde{H}_v$.  We initialize $S_w$ to be empty and then scan the connected components in $\widetilde{H}_v$ in vertical order.  For each component $C$ of $\widetilde{H}_v$, we process the regions in $C$ in topological order as follows.  Let $\tau$ be the region in $C$ being examined.  If $\tau$ does not intersect the interior of $\slab(w)$, ignore it.  Suppose that $\tau$ intersects the interior of $\slab(w)$.  If $\tau$ does not span $v$, then add $\tau \cap \slab(w)$ to $S_w$ as a legal region.  Suppose that $\tau$ spans $v$.  We tentatively add $\tau \cap \slab(w)$ as a candidate illegal region to $S_w$.  Then, we check if $\up(\tau \cap \slab(w)) = \low(\tau')$ for some candidate illegal region $\tau' \in S_w$, and if so, we merge $\tau \cap \slab(w)$ into $\tau'$.

The second stage processes $H_v \setminus \widetilde{H}_v$.  In this stage, more regions may be added to $S_w$.  We also build another set $\overline{S}_w$ of regions, which will become $H_w \setminus \widetilde{H}_w$.  We repeat the following for every component $C$ in $H_v \setminus \widetilde{H}_v$.  Compute the set of connected components in $C \cap \slab(w)$.  Those components in $C \cap \slab(w)$ that do not intersect both bounding lines of $\slab(w)$ are added to $\overline{S}_w$.  For each component $C'$ in $C \cap \slab(w)$ that intersects both bounding lines of $\slab(w)$, we insert $C'$ into $S_w$.  The location of $C'$ in $S_w$ is determined by a search using any intersection between $C'$ and the left bounding line of $\slab(w)$.  Moreover, the legal and candidate illegal regions in $C'$ are generated by a topological sort of $C'$ as described in the first stage.

Finally, we scan $S_w$ to check the candidate illegal regions.  Those that do not have any vertex in the interior of $\slab(w)$ are removed.  The pruned $S_w$ becomes $\widetilde{H}_w$.  The union of $\widetilde{H}_w$ with $\overline{S}_w$ is $H_w$, i.e., $\overline{S}_w = H_w \setminus \widetilde{H}_w$.  The processing time is $O(|H_v|\log n)$.  We are now ready to recurse at $w$ using $H_w$.  

The recursive call at $w$ will return a list, \emph{updated-trees}$(w)$, of the roots of some gap trees of $w$.  Each tree in \emph{updated-trees}$(w)$ stores a maximal contiguous subsequence of free and blocked gaps of $w$ that are induced by a blocked gap $g$ of $v$ affected by the triangulation-update.  
For every gap tree $T$ in \emph{updated-trees}$(w)$, we take an arbitrarily point $x$ covered by the gaps in $T$, find the blocked gap $g \in \gaplist(v)$ that contains $x$, and set a gap tree pointer from $g$ to $T$.

\begin{lemma}
	\label{lem:modify-recurse}
	Preparing for the recursive calls at the children of $v$ takes $O(|H_v|\log n)$ time.  Upon return from the recursions at the children $v_L$, $v_M$ and $v_R$ of $v$, it takes $O(|H_v|\log n)$ time to set gap tree pointers from the appropriate blocked gaps in $\gaplist(v)$ to gap trees in updated-trees$(v_L)$, updated-trees$(v_M)$, and updated-trees$(v_R)$.
\end{lemma}

\subsection{Updating the gap tree at a slab tree node}
\label{sec:update-gaptree}

We need to return a list \emph{updated-trees}$(v)$ of the roots of the gap trees of $v$ that are affected by the triangulation-update.  
The contruction of \emph{updated-trees}$(v)$ goes hand in hand with the execution of {\sc Modify}$(v)$.  Specifically, whenever we execute step~6(b) of {\sc Modify}$(v)$ to replace a subsequence of free and blocked gaps in the current $L_v$ by $\Sigma$, we need to add gap trees to \emph{updated-trees}$(v)$.   At the same time, we maintain the set $\cal F$ of gap trees of $v$ with respect to the current $L_v$.   

We first compute some gap trees for the maximal contiguous subsequences of gaps in $\Sigma$.  Let $(T_1, \cdots, T_k)$ be the resulting trees in order from top to bottom.  It is possible that $T_1 = T_k$.  There are two cases according to step~6(b), depending on whether $\Sigma$ replaces the gaps in $L_v$ from $\hat{g}_i$ to $\check{g}_i$, or from $\hat{g}_i$ to $\low(C_i)$.

Consider the first case.  Let $\hat{T}$ be the gap tree in $\cal F$ that contains $\hat{g}_i$.  Split $\hat{T}$ at $\up(\hat{g}_i)$ into two trees $\hat{T}_1$ and $\hat{T}_2$ that are above and below $\up(\hat{g}_i)$, respectively.  If $\hat{T}_1$ is non-empty, replace the occurrence of $\hat{T}$ in $\cal F$ by $\hat{T}_1$ and $\hat{T}_2$ in this order.  Symmetrically, let $\check{T}$ be the gap tree in the current $\cal F$ that contains $\check{g}_i$.  Split $\check{T}$ at $\low(\check{g}_i)$ into two trees $\check{T}_1$ and $\check{T}_2$ that are above and below $\low(\check{g}_i)$, respectively.  If $\check{T}_2$ is non-empty, replace the occurrence of $\check{T}$ in $\cal F$ by $\check{T}_1$ and $\check{T}_2$ in this order.  Next, replace by $(T_1,\cdots,T_k)$ the gap trees in $\cal F$ from the one containing $\hat{g}_i$ to the one containing $\check{g}_i$.  Finally, if the highest gap in $T_1$ is adjacent to the lowest gap of the gap tree $T$ in $\cal F$ above $T_1$, merge $T$ and $T_1$; if the lowest gap in $T_k$ is adjacent to the highest gap of gap tree $T'$ in $\cal F$ below $T_k$, merge $T_k$ and $T'$. 

The corresponding change to \emph{updated-trees}$(v)$ is as follows.  The gap trees $T_2, \cdots, T_{k-1}$ are inserted into \emph{updated-trees}$(v)$.  If $T$ and $T_1$ are merged, insert the merge of $T$ and $T_1$ at the front of \emph{updated-trees}$(v)$.  If $T$ and $T_1$ are not merged, insert $T_1$ at the front of \emph{updated-trees}$(v)$ and then if $\hat{T}_1$ is non-empty, insert $\hat{T}_1$ at the front of \emph{updated-trees}$(v)$ afterwards.  The handling of $T'$ and $T_k$ is similar.   If $T'$ and $T_k$ are merged, append the merge of $T'$ and $T_k$ to \emph{updated-trees}$(v)$.  If $T'$ and $T_k$ are not merged, append $T_k$ to \emph{updated-trees}$(v)$ and then if $\check{T}_2$ is non-empty, append $\check{T}_2$ to \emph{updated-trees}$(v)$ afterwards.   

The case of $\Sigma$ replacing the gaps in $L_v$ from $\hat{g}_i$ to $\low(C_i)$ is handled similary.  The only difference is that we first split $\check{g}_i$ using $\low(C_i)$, thus modifying the gap tree $\check{T}$ in $\cal F$ that contains $\check{g}_i$, and then split the modified $\check{T}$ using $\low(C_i)$.

Clearly, the total running time is $O(|H_v|\log n)$ plus the time to delete the gap trees in $\cal F$ that are replaced.  Each such deletion takes $O(\log n)$ time, assuming that $\cal F$ is represented by a balanced search tree.  In Section~\ref{sec:rebuild}, we will introduce a periodic rebuild of $\cal T$ and its auxiliary structures so that $n = \Theta(n')$, where $n'$ is the number of triangles in $\Delta_B$ in the initial construction or the last rebuild, whichever is more recent.  The gap tree deleted from $\cal F$ might be inserted in the past since the initial construction or the last rebuild, or in the initial construction or the last rebuild, whichever is applicable and more recent.  We charge the $O(\log n)$ deletion time to the insertion of that deleted gap tree.  Note that we might have spent as little as $O(\log n')$ time in inserting that gap tree into $\cal F$ in the past.  Nevertheless,  $\log n = \Theta(\log n')$ as $n = \Theta(n')$, and therefore, the charging argument goes through.  Thus, the total running time is $O(|H_v|\log n)$ amortized.

\begin{lemma}
	\label{lem:modify-amort}
	The gap trees of $v$ can be updated in $O(|H_v|\log n)$ amortized time.
\end{lemma}

\subsection{Periodic rebuild}
\label{sec:rebuild}

Lemmas~\ref{lem:modify-time}--\ref{lem:modify-amort} show that the update time is $O(|H_v|\log n)$ amortized, provided that $n = \Theta(n')$, where $n'$ is the number of vertices in the initial construction or the last rebuild, whichever is more recent.  To enforce this assumption, we need to rebuild the slab tree $\cal T$ and its auxiliary structures periodically.  Let $c < 1/2$ be a constant.  We rebuild $\cal T$ and its auxiliary structures with respect to $\cal L$ and the current $\Delta_B$ when the total size of triangulation-updates exceeds $cn'$ since the initial construction or the last rebuild, where $n'$ was the number of triangles in $\Delta_B$ then.

\begin{lemma}
\label{lem:1}
Let $n$ denote the number of triangles in $\Delta_B$.
\begin{itemize}
	\item $n = \Theta(n')$, where $n'$ is the number of triangles in $\Delta_B$ in the initial construction or the last rebuild, whichever is more recent.
	\item Any query can be answered in $O(\mathrm{OPT} + \log \log n)$ expected time, where {\em OPT} is the minimum expected query time of the best point location decision tree for $\Delta$.
	\item The data structure uses $O(n\log n)$ space and can be constructed in $O(|{\cal L}|\log n)$ time.
	\item A triangulation-update of size $k \leq n/2$ takes $O(k\log^2 n + (|{\cal L}|\log n)/n)$ amortized time.
\end{itemize}
\end{lemma}
\begin{proof}
Let $U$ denote a triangulation-update.  Recall that $T_U$ is the new triangulation of the polygonal region $R_U$ affected by $U$.  Both $|R_U|$ and $|T_U|$ are $O(k)$, and $T_U$ can be constructed in $O(k\log k)$ time.  The periodic rebuild ensures that $n \leq n' + cn' + k \leq (c+1)n' + n/2 \Rightarrow n \leq 2(c+1)n'$ and $n \geq n'- cn' - k \geq (1-c)n' - n/2 \Rightarrow n \geq 2(1-c)n'/3$.  It follows that $n = \Theta(n')$.
	
The query time bound follows from Lemma~\ref{lem:dyntrianalysis}.  The space and preprocessing time follow from Lemma~\ref{lem:slabtree}.  The correctness of the update follows from the discussion in Sections~\ref{sec:update-gaplist}--\ref{sec:rebuild} and Appendix~\ref{app:modify}.  It remains to bound the amortized update time.  

By Lemmas~\ref{lem:modify-time}--\ref{lem:modify-amort}, the update time is $O(\sum_v |H_v| \log n)$ amortized.  Each legal region in $H_v$ is part of a triangle $t \in T_U$ that intersects the interior of $\slab(v)$ but does not span $\mathit{parent}(v)$.  The depth of the slab tree is $O(\log n)$.  Therefore, as in the case of a segment tree, $t$ is stored as $O(\log n)$ legal regions at $O(\log n)$ nodes of the slab tree.  This contributes a term of $O(k\log n)$ to $\sum_v |H_v|$.  Each illegal region $\tau$ in $H_v$ has a boundary vertex, say $x$, in the interior of $\slab(v)$.  Also, the complexity of $\tau$ is $O(1)$.   We charge the complexity of $\tau$ to $x$.  At the node $v$, we cannot charge the complexity of another illegal region in $\slab(v)$ to $x$.  Since the slab tree has $O(\log n)$ depth, $x$ is charged $O(\log n)$ times.  Only vertices of $T_U$ can be charged.  It follows that the total complexity of all illegal regions in all $H_v$'s is $O(k\log n)$.  This allows us to conclude that $\sum_v |H_v| = O(k\log n)$, implying that the update time is $O(k\log^2 n)$ amortized.  The periodic rebuilding of $\cal T$ adds another $O((|{\cal L}|\log n)/n)$ amortized time.

We also need to update the backup worst-case dynamic planar point location structure.  The triangulation-update $U$ can be formuated as a seuqence of $O(k)$ edge deletions and edge insertions.  We can use any one of the data structures in~\cite{paper:Arge2006, paper:Chan2015, paper:Chiang1992, paper:Preparata1989} to represent the  backup structure.  Each edge insertion or deletion can be done in $O(\log^2 n)$ time.
\end{proof}

\section{Allowing arbitrary vertex location}
\label{sec:general}

In this section, we discuss how to allow a new vertex to appear anywhere instead of on one of the fixed lines in $\cal L$.  This requires revising the slab tree structure.  The  main issue is how to preserve the geometric decrease in the probability of a query point falling into the slabs of internal nodes on every root-to-leaf path in $\cal T$.

Initialize $\cal L$ to be the set of vertical lines through the vertices of the initial $\Delta_B$.  Construct the initial slab tree $\cal T$ for $\Delta_B$ and $\cal L$ using the algorithm in Section~\ref{sec:construct}.  Whenever $\cal T$ is rebuilt, we also rebuild $\cal L$ to be the set of vertical lines through the vertices of the current $\Delta_B$.   
%
Between two successive rebuilds, we grow $\cal L$ monotonically as triangulation-updates are processed.   Although every vertex of $\Delta_B$ lies on a line in $\cal L$, some line in $\cal L$ may not pass through any vertex of $\Delta_B$ between two rebuilds.

The free, blocked, and shadow gaps of a slab tree node are defined as in Section~\ref{sec:dyntridatastructure}.  So are the gap trees of a slab tree node.  However, gap weights are redefined in Section~\ref{sec:gap-weight} in order that they are robust against small geometric changes.

When a triangulation-update $U$ is processed, 
we first process the vertical lines through the vertices of $T_U$ before we process $T_U$ as specified in Section~\ref{sec:update}.   For each vertical line $\ell$ through the vertices of $T_U$, if $\ell \not\in {\cal L}$, we insert $\ell$ into $\cal L$ and then into $\cal T$.  To allow for fast line insertion into $\cal T$, we increae the number of children of an internal node to $O(\log n)$, and we need to classify the children appropriately.
Sections~\ref{sec:child} and~\ref{sec:insert-line} provide the details of this step.  The processing of $T_U$ is discussed in Section~\ref{sec:update2}.


Querying is essentially the same as in Section~\ref{subsec:loglognanalysis} except that we need a fast way to descend the slab tree as some nodes have $O(\log n)$ children.  This is described in Section~\ref{sec:child}. 

\subsection{Weights of gaps and more}
\label{sec:gap-weight}

Let $n'$ be the number of triangles in $\Delta_B$ at the time of the initial construction or the last rebuild of $\cal T$, whichever is more recent.  Let $N = 2(c+1)n'$, where $c$ is the constant in the threshold $cn'$ for triggering a rebuild of $\cal T$.

For every free gap $g$, let $t_g$ denote the triangle in the current $\Delta_B$ that contains $g$, and we define the weight of $g$ to be $\wt(g) = \max\left\{\Pr(t_g), 1/N\right\}$.   The alternative $1/N$ makes the access time of $g$ in a gap tree no worse than $O(\log N) = O(\log n)$.

\cancel{
What is the reason for setting $\wt(g)$ this way?  
By rebuiding, the number $n$ of triangles in $\Delta_B$ is at most $N$, provided that each triangulation-update has size at most $n/2$.  
We will see that a node in $\cal T$ has at most $2\log_2 N + 3$ children.   A weight of $\frac{\Pr(t_g)}{2\log_2 N + 3}$ ensures that the total weight of free gaps induced by $t_g$ among the children of a node is at most $\Pr(t_g)$.  This is important in the query time analysis.  The alternative $\frac{1}{N(2\log_2 N + 3)}$ 
makes the access time of a biased search tree storing $g$ no worse than $O(\log n)$.
}

For every blocked gap $g$, every vertex $p$ of $\Delta_B$, and every node $v$ of $\cal T$, define:

\begin{itemize}
	\item $\wt(p)$ = sum of  $\max\left\{\Pr(t), 1/N\right\}$ over all triangles $t \in \Delta_B$ incident to $p$.
	\item $\mathit{vert}(g) = \{ \text{vertex $p$ lying in $g$} : \exists\, \text{triangle}\, pqr \in \Delta_B \,\, \text{s.t.}\,\, \mathit{interior}(pqr) \cap \mathit{interior}(g) \not= \emptyset \}$.
	\item $\wt(g)  = \sum_{p \in \mathit{vert}(g)} \wt(p)$.
	\item \emph{blocked-gaps}$(p)$ = $\{\text{blocked gap $g$} : p \in \mathit{vert}(g)\}$.
	\item $\mathit{vert}(v)$ = the subset of vertices of $\Delta_B$ that lie in $\slab(v)$.
	\item $\mathit{lines}(v)$ =  the subset of lines in $\cal L$ that intersect $\slab(v)$.
\end{itemize}

The set $\mathit{vert}(g)$ is only used for notational convenience.  The set \emph{blocked-gaps}$(p)$ is not stored explicitly.  We discuss how to retrieve \emph{blocked-gaps}$(p)$ in Section~\ref{sec:child}.  The sets
$\mathit{vert}(v)$ and $\mathit{lines}(v)$ are stored as balanced search trees in increasing order of $x$-coordinates.

\subsection{Revised slab tree structure}  
\label{sec:child}

\paragraph*{Node types}  A vertical line \emph{pierces} a slab if the line intersects the interior of that slab.  An internal node $v$ of $\cal T$ has children of two possible types.

\begin{itemize}
	
	\item {\sc Heavy-child:}~A child $w$ of $v$ is a \emph{heavy-child} if $\Pr(\slab(w)) > \Pr(\slab(v))/2$.   
	\begin{itemize}
		
		\item The heavy-child $w$ may be labelled \emph{active} or \emph{inactive} upon its creation.  This label will not change.  If $w$ was created in the initial construction or a rebuild of $\cal T$, then $w$ is inactive.  
		
		\item If $w$ is inactive, $\gaplist(w)$ and the gap trees of $w$ are represented as before.  If $w$ is active, then $w$ is a leaf, and $\gaplist(w)$ and the gap trees of $w$ are stored as persistent data structures using the technique of node copying~\cite{driscoll89}.   
	\end{itemize}
	
	\item {\sc Light-child:}~There are two sequences of \emph{light-children} of $v$, denoted by \emph{left-light}$(v)$ and \emph{right-light}$(v)$, which satisfy the following properties.
	\begin{itemize}
		\item For each light child $w$ of $v$, $\Pr(\slab(w)) \leq \Pr(\slab(v))/2$.
		
		\item For each light child $w$ of $v$, $\gaplist(w)$ and the gap trees of $w$ are represented as before.
		
		\item Let \emph{left-light}$(v)$ = $(w_1, w_2, \cdots, w_k)$ and let \emph{right-light}$(v)$ = $(w_{k+1}, w_{k+2}, \cdots, w_m)$ in the left-to-right order of the nodes.
		\begin{itemize}
			\item For $i \in [1,k-1] \cup [k+1,m-1]$, $\slab(w_i)$ and $\slab(w_{i+1})$ are interior-disjoint and share a boundary.
			\item If $v$ has an active heavy-child $w$, then $\slab(w)$ is bounded by the right and left boundaries of $\slab(w_k)$ and $\slab(w_{k+1})$, respectively.  Otherwise, the right boundary of $\slab(w_k)$ is the left boundary of $\slab(w_{k+1})$.
			\item If $v$ does not have an active heavy child, $v$ has at most $2\log_2 N + 2$ children.
			\item If $v$ has an active heavy-child, the following properties are satisfied.
			\begin{romani}
				\item For $r \geq 1$, a light-child $w$ of $v$ has \emph{rank} $r$ if the number of lines in $\cal L$ that intersect $\slab(w)$ is in the range $[2^r,2^{r+1})$.  So $r \leq \log_2 N$, where $N = 2(c+1)n'$.  We denote $r$ by $\mathit{rank}(w)$.  
				\item We have $\mathit{rank}(w_1) > \cdots > \mathit{rank}(w_k)$ and $\mathit{rank}(w_{k+1}) < \cdots < \mathit{rank}(w_m)$.  For $r \in [1,\log_2 N]$, there is at most one light-child of rank $r$ in each of \emph{left-light}$(v)$ and \emph{right-light}$(v)$ .
			\end{romani}
			
		\end{itemize}

	\end{itemize}
	
\end{itemize}


Since an internal node has $O(\log n)$ children, each triangle $t \in \Delta_B$ induces $O(\log n)$ free gaps at each level of the slab tree $\cal T$, resulting in $O(\log^2 n)$ free gaps.  Each vertex of $t$ may also contribute to the boundary complexity of at most $O(\log n)$ blocked gaps at slab tree nodes whose slabs contain that vertex of $t$.  As a result, the revised slab tree and its auxiliary structures take up $O(n\log^2 n)$ space.

\paragraph*{Node access}  
Each node $v$ keeps a biased search tree \emph{children}$(v)$.  The weight of a child $w$ in \emph{children}$(v)$ is  $\max\bigl\{\frac{\Pr(\slab(v))}{2\log_2 N + 2}, \, \Pr(\slab(w))\bigr\}$, where $N = 2(c+1)n'$.   Since $n  = \Theta(n')$, accessing
$w$ takes $O\bigl(\min\bigl\{\log \frac{\Pr(\slab(v))}{\Pr(\slab(w))}, \, \log\log n \bigr\} \bigr)$ time.

For each blocked gap $g$ of $v$, we use a biased search tree $T_g$ to store pointers to the gap trees induced by $g$ at the children of $v$. The weight of the node in $T_g$ that represents a gap tree $T$ at a child $w$ is $\max\bigl\{\frac{\Pr(\slab(v))}{2\log_2 N + 2}, \, \Pr(\slab(w))\bigr\}$.  Accessing $T$ via $T_g$ takes  $O\bigl(\min\bigl\{\log \frac{\Pr(\slab(v))}{\Pr(\slab(w))}, \, \log\log n \bigr\} \bigr)$ time.

Given a vertex $p$ of $\Delta_B$, there are $O(\log n)$ blocked gaps in \emph{blocked-gaps}$(p)$ and we can find them as follows.  Traverse the path from the root of $\cal T$ to the leaf whose slab contains $p$, and for each node $v$ encountered, we search in $\gaplist(v)$ to retrieve the blocked gap of $v$, if any, that contains $p$.  The time needed is $O(\log^2 n)$.



\subsection{Insertion of a vertical line into the slab tree}
\label{sec:insert-line}

Let $\ell$ be a new vertical line.  We first insert $\ell$ into $\cal L$ and then  insert $\ell$ into $\cal T$ in a recursive traversal towards the leaf whose slab is pierced by $\ell$.  

\paragraph*{Internal node}  Suppose that we visit an internal node $v$.  We first insert $\ell$ into $\mathit{lines}(v)$.  We query \emph{children}$(v)$ to find the child slab pierced by $\ell$.  If $v$ does not have an active heavy-child, recursively insert $\ell$ at the child found.  Otherwise, we work on \emph{left-light}$(v)$ or \emph{right-light}$(v)$ as follows.

\begin{itemize}
	
	\item Case~1: $\ell$ pierces $\slab(w_j)$ for some $w_j \in \text{\emph{left-light}}(v)$.  If $\slab(w_j)$ intersects fewer than $2^{\mathit{rank}(w_j)+1}$ lines in $\cal L$, recursively insert $\ell$ into $w_j$ and no further action is needed.\footnote{The line $\ell$ has already been inserted into $\cal L$.}  Otherwise, $\slab(w_j)$ intersects $2^{\mathit{rank}(w_j)+1}$ lines in $\cal L$, violating the structural property of a light-child.  In this case, we merge some nodes in \emph{left-light}$(v)$ as follows.  
	
	\hspace*{12pt}Let \emph{left-light}$(v) = (w_1,\cdots,w_j,\cdots)$.  Find the largest $i \leq j$ such that  the number of lines in $\cal L$ that intersect $\slab(w_i) \cup \cdots \cup \slab(w_j)$ is in the range $[2^r, 2^{r+1})$ for some $\mathit{rank}(w_i) \leq r < \mathit{rank}(w_{i-1})$.\footnote{If $i = 1$, treat $\mathit{rank}(w_{i-1})$ as $\infty$.}  Note that $r > \mathit{rank}(w_j)$.  Let $\ell_L$ denote the left boundary of $\slab(w_i)$.  Let $\ell_R$ denote the right boundary of $\slab(w_j)$.  Let $S$ denote the slab bounded by $\ell_L$ and $\ell_R$.  We rebuild the slab subtree rooted at $w_i$ and its auxiliary structures to expand $\slab(w_i)$ to $S$ as follows.   It also means that $\mathit{rank}(w_i)$ is updated to $r$.  The children $w_{i+1},\ldots,w_j$ and their old subtrees are deleted afterwards.
	
	\hspace*{12pt}Let $V = \mathit{vert}(w_i) \cup \cdots \cup \mathit{vert}(w_j)$.
	Let $\Lambda = \mathit{lines}(w_i) \cup \cdots \cup \,\mathit{lines}(w_j)$.  First, we construct a new slab subtree rooted at $w_i$ with respect to $V$ and $\Lambda$ as described in Section~\ref{sec:struct}.  No auxiliary structure is computed yet.  
	We control the construction so that it does not produce any node at depth greater than $\log_2 N = O(\log n)$ with respect to the whole slab tree.  The construction time is $O(|\Lambda|\log n)$.  Afterwards, $\slab(w_i)$ becomes $S$.  {\em Label all heavy-children in the new slab subtree rooted at $w_i$ as inactive.}
	
	\hspace*{12pt}
	Mark the triangles that are incident to the vertices in $V$, overlap with $S$, and do not span $S$.  Let $G_{w_i}$ be the set of marked triangles.  This takes $O(|G_{w_i}|)$ time, assuming that each vertex $p$ has pointers to its incident triangle(s) intersected by a vertical line through $p$.
	The old blocked gaps of $w_i$ will be affected by the rebuild at $w_i$.  The old free gaps of $w_i$ contained in some triangles in $G_{w_i}$ will be absorbed into some blocked gaps after the rebuild.  The other old free gaps of $w_i$ are not affected because their containing triangles span $S$.  
	
	\hspace*{12pt}To update $\gaplist(w_i)$, intersect $G_{w_i}$ with $S$ to generate the directed acyclic graph $H_{w_i}$ and then update $\gaplist(w_i)$ as in Section~\ref{sec:update}.  This takes $O(|G_{w_i}|\log n)$ amortized time.  Only the blocked gaps of $w_i$ can induce gap lists and gap trees at the descendants of $w_i$.  Therefore, as in the construction algorithm in Section~\ref{sec:construct}, we can take the subset of $G_{w_i}$ that induce the blocked gaps of $w_i$ and recursively construct the gap lists and gap trees at the descendants of $w_i$.  This takes $O(|G_{w_i}|\log^2 n)$ time by an analysis analogous to the one for Lemma~\ref{lem:slabtree}.\footnote{Since $w_i$ has $O(\log n)$ instead of $O(1)$ children, the construction time has an extra log factor.}  For each blocked gap $g$ of $w_i$, we create a biased search tree of pointers to the gap trees induced by $g$ at the children of $w_i$.  
	
	\hspace*{12pt}The update of $\gaplist(w_i)$ preserves the old shadow gaps of $w_i$, and it does not generate any new shadow gap of $w_i$.  Therefore, no two gap trees of $w_i$ can be merged and no gap tree of $w_i$ can be split, although the content of a gap tree may be modified.  A gap tree of $w_i$ is updated only  when some free gaps in it are merged into some blocked gaps.  Thus, updating the gap trees of $w_i$ takes $O(|G_{w_i}|\log n)$ time.
	
	\hspace*{12pt}Finally, $\mathit{vert}(w_i) := V$, $\mathit{lines}(w_i) := \Lambda$, and the recursive insertion of $\ell$ terminates.

	\item Case~2: $\ell$ pierces $\slab(w_j)$ for some $w_j \in \text{\emph{right-light}}(v)$.  Symmetric to Case~1.
	
	\item Case~3: $\ell$ pierces the active heavy-child of $v$.  An active heavy-child is a leaf of the slab tree.  We discuss how to insert a vertical line at a leaf next.
	
\end{itemize}

\paragraph*{Leaf node}  Suppose that we come to a leaf $v$.  If $\text{depth}(v) = \log_2 N$, do nothing and return.  Otherwise, there are two cases.  Note that $\gaplist(v)$ consists of free gaps only.  The line $\ell$ divides $\slab(v)$ into slabs $S_L$ and $S_R$ on the left and right of $\ell$, respectively. 

\begin{itemize}
	
	\item Case~1: $v$ is not an active heavy-child of $\mathit{parent}(v)$.   Turn $v$ into an internal node by making two children $w_L$ and $w_R$ of $v$ with $\slab(w_L) = S_L$ and $\slab(w_R) = S_R$.  If $\Pr(\slab(w_L)) > \Pr(\slab(v))/2$, then $w_L$ is the heavy-child of $v$, label $w_L$  active, and set \emph{left-light}$(v):= \emptyset$.  If not, $w_L$ is a light-child of rank one and set \emph{left-light}$(v):= (w_L)$.   The handling of $w_R$ is symmetric.  As $\gaplist(v)$ consists of free gaps only, $\gaplist(w_L)$ and $\gaplist(w_R)$ are empty.  So $w_L$ and $w_R$ have no gap tree.  The initializations of $\mathit{vert}(w_L)$, $\mathit{vert}(w_R)$, $\mathit{lines}(w_L)$, and $\mathit{lines}(w_R)$ are trivial.  
	
	\item Case~2: $v$ is an active heavy-child of $\mathit{parent}(v)$.   We expand \emph{left-light}$(\mathit{parent}(v))$ and/or \emph{right-light}$(\mathit{parent}(v))$ as follows.  \emph{}W.l.o.g., assume that $\Pr(S_L) \leq \Pr(\slab(v))/2$.  Update $\slab(v):= S_R$, which does not change $\gaplist(v)$ or any gap tree of $v$ combinatorially.    The weights of gaps in $\gaplist(v)$ are also unaffected.\footnote{This is the reason why gap weights have to be redefined in order to be robust against small geometric changes.}
	\begin{itemize}
		\item Case~2.1: $\Pr(S_R) \leq \Pr(\slab(\mathit{parent}(v)))/2$.  Then, $\mathit{parent}(v)$ has no heavy-child afterwards.  Note that $\mathit{parent}(v)$ has at most $2\log_2 N + 1$ children before this update, where $N = 2(c+1)n'$.  Create a light-child $w_L$ of $\mathit{parent}(v)$ with $\slab(w_L) = S_L$.   Note that $\gaplist(w_L)$ and the gap trees of $w_L$ are combinatorially identical to those of $v$, which are stored as persistent search trees.   We copy them to form $\gaplist(w_L)$ and the gap trees of $w_L$, each taking $O(1)$ amortized space and time.\footnote{This explains why persistent data structures are used for the gap list and gap trees of an active heavy child.}  Append $w_L$ to \emph{left-light}$(\mathit{parent}(v))$.  Add $v$ to \emph{right-light}$(\mathit{parent}(v))$ as its leftmost element.  Therefore, $\mathit{parent}(v)$ has at most $2\log_2 N + 2$ children afterwards.
		
		\item Case~2.2: $\Pr(S_R) >\Pr(\slab(\mathit{parent}(v)))/2$.  Then, $v$ remains the active heavy-child of $\mathit{parent}(v)$.  We handle $S_L$ as follows.
		\begin{itemize}
			\item If  \emph{left-light}$(\mathit{parent}(v))$ contains no light-child of rank one, then create a light-child $w_L$ with $\slab(w_L) = S_L$, initialize $\gaplist(w_L)$ and the gap trees of $w_L$ as in Case~2.1 above, and append $w_L$ to \emph{left-light}$(\mathit{parent}(v))$.
	
		  	\item Otherwise, let \emph{left-light}$(\mathit{parent}(v)) = (w_1,\cdots,w_k)$, i.e., $\mathit{rank}(w_k) = 1$.  Find the largest $i \leq k$ such that  the number of lines in $\cal L$ that intersect $\slab(w_i) \cup \cdots \cup \slab(w_k) \cup S_L$ is in the range $[2^r, 2^{r+1})$ for some $\mathit{rank}(w_i) \leq r < \mathit{rank}(w_{i-1})$. Expand $\slab(w_i)$ to the slab bounded by the left boundary of $\slab(w_i)$ and the right boundary of $S_L$ as in Case~1 of the insertion of a vertical line at an internal node.  Rebuild the slab subtree rooted at $w_i$ and its auxiliary structures, again as in Case~1 of the insertion of a vertical line at an internal node.  \emph{Label all heavy-children in the new slab subtree rooted at $w_i$ as inactive.}
		  \end{itemize}
		
		\end{itemize}
\end{itemize}

We show that the above line insertion algorithm has a good amortized running time.  The key idea is that a rebuild at a light-child $w_i$ must be triggered by the insertions of many new lines.  As the heavy-children in the new slab subtree after the rebuild are inactive, these lines cannot be charged by any rebuild at the descendants of $w_i$ in the future.

\begin{lemma}
	\label{lem:insert2}
	Let $\cal T$ be the slab tree constructed for a set ${\cal L}'$ of vertical lines and $\Delta_B$ in the initial construction or the last rebuild, whichever is more recent.  For any ${\cal L} \supset {\cal L}'$, the insertion time of lines in ${\cal L} \setminus {\cal L}'$ into $\cal T$ is $O(|{\cal L} \setminus {\cal L}'|\log^4 n)$ plus some charges on edges of $\Delta_B$ such that every edge gains at most $O(\log^4 n)$ charge since the initial construction or the last rebuild, whichever is more recent.
\end{lemma}
\begin{proof}
	Consider the insertion of a line $\ell \in {\cal L} \setminus {\cal L}'$.  There are two components in the insertion time.  First, the time to go from the root of $\cal T$ to some node whose slab is pierced by $\ell$ and where the recursive insertion terminates.   Second, the time to perform the structural changes at the nodes on this path of $\cal T$ traversed. 

Let $v_1, \cdots, v_m$ be nodes of $\cal T$ visited, where $v_1$ is the root of $\cal T$ and $v_m$ is the last node visited.  It takes $O(\log\log n)$ time to visit $v_{k+1}$ via \emph{children}$(v_k)$.   As the depth of $\cal T$ is $O(\log n)$, summing over $k$ gives $O(\log n \log\log n)$.  

The updating algorithm spends $O(1)$ time at almost every node from $v_1$ to $v_m$, except that there may be one node $v_k$ (which is $v_{m-1}$ or $v_m$) such that some nodes in \emph{left-light}$(v_k)$ or \emph{right-light}$(v_k)$, say \emph{left-light}$(v_k)$, are merged.  Let $w_i, \cdots, w_j$ be the node in \emph{left-light}$(v_k)$ that are merged.  The merging  takes $O((|\Lambda_k| + |E_k|)\log^2 n)$ time, where $\Lambda_k = \mathit{lines}(w_i) \cup \cdots \cup  \, \mathit{lines}(w_j)$ and $E_k$ is the set of edges in $\Delta_B$ with at least one endpoint in the interior of $\slab(w_i) \cup \cdots \cup \slab(w_j)$.  Let $r = \lfloor \log_2 |\Lambda_k|\rfloor$.   So $|\Lambda_k| \in [2^r,2^{r+1})$.   

We first analyze the term $O(|\Lambda_k|\log^2 n)$.  There are two cases.
\begin{itemize}
	\item Case~1:~$r > \mathit{rank}(w_i)$.  We pay for the term $O(|\Lambda_k|\log^2 n)$ by charging $O(\log^2 n)$ to each line in $\Lambda_k$.
	
	\item Case~2:~$r = \mathit{rank}(w_i)$. It must be the case that $|\mathit{lines}(w_{i+1}) \cup \cdots \cup \mathit{lines}(w_j)| \geq 2^r$.  We pay for the term $O(|\Lambda_k|\log^2 n)$ by charging $O(\log^2 n)$ to each line in $\mathit{lines}(w_{i+1}) \cup \cdots\cup \mathit{lines}(w_j)$.  
\end{itemize}
In either Case~1 or~2, each line that gains an $O(\log^2 n)$ charge is moved to a node in \emph{left-light}$(v_k)$ of a higher rank.  For the merging to happen, $v_k$ must have an active heavy-child.  Moreover, ${\cal L} \setminus {\cal L}'$ contains almost all lines in $\Lambda_k$ except possibly the left bounding line of $\slab(w_i)$ and the right bounding line of $\slab(w_j)$.  The merging is achieved by rebuilding the slab subtree $T$ rooted at $w_i$ and its auxiliary structures to cover all lines in $\Lambda_k$.  All heavy-children in $T$ are labelled inactive afterwards.   As a result, lines in $\Lambda_k$ in the descendants of $w_i$ in $T$ will not be moved in the future until $T$ is destroyed by another rebuilding.  New vertical lines may be inserted into $T$ after the rebuild, and these lines may move from a light-child in $T$ to a sibling light-child of higher rank in the future.  But no those in $\Lambda_k$.  A line is moved at most $O(\log n)$ times at each level of $\cal T$ because the ranks of light-children range from 1 to at most $\log_2 N = O(\log n)$.  Hence, the total charge at each line in ${\cal L} \setminus {\cal L}'$ is $O(\log^4 n)$.

Consider the other term $O(|E_k|\log^2 n)$.  If we are in case~1 above, an endpoint of every edge in $E_k$ moves to a node in \emph{left-light}$(v_k)$ of a higher rank.  If we are in case~2 above, there must be $2^{r+1}-1$ vertical lines that interesct $\slab(w_i)$ after merging because $|\mathit{lines}(w_i)| \geq 2^r$ before merging and $|\mathit{lines}(w_{i+1}) \cup \cdots \cup \mathit{lines}(w_j)| \geq 2^r$.  It means that if we rebuild the slab subtree rooted at $w_i$ again in the future, $\mathit{rank}(w_i)$ must increase.  Hence, we can charge $O(|E_k|\log^2 n)$ to edges in $E_k$, and each edge in $E_k$ gains an $O(\log^2 n)$ charge at most twice before an endpoint of it is moved to a node of a higher rank.  The edge endpoint that is moved to a node of a higher rank lies on a line in $\Lambda_k$, which means that this movement coincides with moving a line in $\Lambda_k$ to a node of a higher rank.  As explained in the previous paragraph, lines in $\Lambda_k$ will not be moved in the future in the descendants of the new slab subtree rooted at $w_i$.  Therefore, edges in $E_k$ will not be charged again for subtree rebuildings at descendants of $w_i$ in the future.  We can then argue as in the previous paragraph that each edge in $E_k$ gains no more than $O(\log^4 n)$ charge since the initial construction or the last rebuild.
\end{proof}

\subsection{Handling triangulation-updates}
\label{sec:update2}

Let $U$ be a triangulation-update of size $k \leq n/2$.  
%
Let $n'$ be the number of triangles in $\Delta_B$ in the initial construction or the last rebuild, whichever is more recent.  Let $c$ be a constant less than 1/2.  If the threshold $cn'$ has been exceeded by the total size of triangulation-updates (including $U$) since the initial construction or the last rebuild, we rebuild $\cal T$ and its auxiliary structures.
It takes $O(n\log^2 n)$ time and space.  If $U$ does not trigger a rebuild, we proceed as follows instead.

\paragraph*{Step~1} Check the $O(k)$ vertical lines through the vertices of $T_U$.   For each line that does not appear in $\cal L$, we insert it into $\cal L$ and then into $\cal T$ as discussed in Section~\ref{sec:insert-line}.  

\paragraph*{Step~2} The weights of $O(k)$ vertices may change and $O(k)$ vertices may be inserted or deleted.   It is straightforward to update the weights of existing vertices, set the weights of new vertices, and delete vertices in $O(k)$ time.  For every vertex $p$ of the old triangulation, let $\mathit{wt'}(p)$ be its weight in the old triangulation.  For every vertex $p$ of the new triangulation, let $\wt(p)$ be its weight in the new triangulation.  
We perform the following action.
\begin{quote}
Action-I:~for every vertex $p$ of the old triangulation that lies in $R_U$,
\begin{itemize}
	\item for every gap $g \in \text{\emph{blocked-gaps}}(p)$, update $\wt(g) := \wt(g) - \mathit{wt'}(p)$;
	\item if $p$ does not lie in the boundary of $R_U$, then for each slab tree node $v$ such that $p \in \mathit{vert}(v)$, delete $p$ from $\mathit{vert}(v)$.
\end{itemize}
\end{quote}
Action-I runs in $O(k\log^2 n)$ time.


\paragraph*{Step~3} For every vertex $p$ of $T_U$ that lies strictly inside $R_U$, and every ancestor $v$ of the leaf node of $\cal T$ whose slab contains $p$, insert $p$ into $\mathit{vert}(v)$.  This step takes $O(k\log^2 n)$ time.

\paragraph*{Step~4} To update the gap lists and the gap trees, traverse $\cal T$ as in Section~\ref{sec:update}.  
For each node $v$ of $\cal T$ visited, form a directed acyclic graph $H_v$ of regions to update $v$ as in Section~\ref{sec:update}.  This step takes $O(\sum_v |H_v|\log n)$ amortized time.

\paragraph*{Step~5} The weight of a free gap does not change as long as its defining triangle is preserved.  The weights of some blocked gaps may not be updated completely yet, and we fix them by performing Action-II below.
Assume that a zero weight is assigned initially to every blocked gap that is created by the triangulation-update and contains vertices in $T_U$ only.
If $g$ is a blocked gap that contains vertices outside $R_U$ and vertices inside $R_U$, the contributions of $\wt(p) = \wt'(p)$ for all vertices $p \in \mathit{vert}(g)$ outside $R_U$ are inherited from the old blocked gap that contains these vertices, whereas the contributions of $\wt'(p)$ for all vertices $p$ in the old triangulation that lies inside $g$ have been taken out in Action-I above.  The contributions of $\wt(p)$ for all new vertices $p \in \mathit{vert}(g)$ inside $T_U$ are to be added to $\wt(g)$ in Action-II below.

\begin{quote}
	Action-II: for each vertex $p$ of $T_U$ and every gap $g \in \text{\emph{blocked-gaps}}(p)$, update $\wt(g) := \wt(g) + \wt(p)$.
\end{quote}
Action-II runs in $O(k\log^2 n)$ time.  


\begin{theorem}
	\label{thm:3}
	Let $n$ denote the number of triangles in $\Delta_B$.
	\begin{itemize}
	\item Any query can be answered in $O(\mathrm{OPT} + \log \log n)$ expected time, where {\rm OPT} is the minimum expected query time of the best point location decision tree for $\Delta$.
	\item The data structure uses $O(n\log^2 n)$ space, and it can be constructed in $O(n\log^2 n)$ time.
	\item A triangulation-update of size $k \leq n/2$ takes $O(k\log^4 n)$ amortized time.
	\end{itemize}
\end{theorem}
\begin{proof}
The space usage follows from the discussion in Section~\ref{sec:child}. 

Given a triangulation-update $U$ of size $k$, Lemma~\ref{lem:insert2} tells us that if no rebuilding of $\cal T$ is triggered, then inserting the vertical lines through the vertices of $T_U$ takes $O(k\log^4 n)$ amortized time, and then the rest of the update takes $O(\sum_v |H_v|\log n + k\log^2 n)$ time as discussed in Section~\ref{sec:update2}.  The sum $\sum_v |H_v|\log n$ is over all slab tree nodes visited during the triangulation-update.  The proof of Lemma~\ref{lem:1} applies to the case when each slab tree node has at most three children, and it shows that $\sum_v |H_v| = O(k\log n)$ in that case.  The same analysis carries over to the case of having $O(\log n)$ children per slab tree node, resulting in $\sum_v |H_v| = O(k\log^2 n)$.  Therefore, a triangulation-update runs in $O(k\log^4 n)$ amortized time, provided that a rebuilding of $\cal T$ is not triggered.  If the rebuilding of $\cal T$ is triggered, the entire slab tree together with all auxiliary structures are rebuilt in $O(n\log^2 n)$ time.  This cost can be amortized over the $\Omega(n')$ changes in $\Delta_B$ since the last rebuild of $\cal T$, resulting in an $O(k\log^2 n)$ amortized update time.

We want to repeat the analysis of the expected query time in the proof of Lemma~\ref{lem:dyntrianalysis}.  However, we no longer have the property that for each internal node $v$ of $\cal T$, $\Pr(\slab(v)) \leq \Pr(\slab(\text{\emph{parent}}(v)))/2$ due to one possible exception---an inactive heavy-child $v$ of $\mathit{parent}(v)$ changes from a leaf to an internal node due to the insertion of a vertical line to $\cal T$.  In this case, after $v$ becomes an internal node, any light-child $w$ of $v$ satisfies $\Pr(\slab(w)) \leq \Pr(\slab(v))/2 \leq \Pr(\slab(\text{\emph{parent}}(v)))/2$.  Moreover, if there is a heavy-child $w'$ of $v$, then $w'$ must be active and hence a leaf.  Therefore, we have the weaker property that for each internal node $u$ of $\cal T$ such that $u' = \mathit{parent}(\mathit{parent}(u))$ exists, then $\Pr(\slab(u)) \leq \Pr(\slab(u'))/2$.  This weaker property implies that for any node $v$ of $\cal T$, $\Pr(\slab(v))  \leq 2^{-\lfloor (\text{depth}(v)-1)/2 \rfloor}$, which implies that $\text{depth}(v) = O\left(-\log \Pr(\slab(v))\right)$.  This allows us to apply the proof of Lemma~\ref{lem:dyntrianalysis} with some modifications:
\begin{itemize} 
	\item In the proof of Lemma~\ref{lem:dyntrianalysis}, the fact that any triangle in $\Delta_B$ may appear as $\beta$ fragments in $\cal T$ leads to an additive term of $O(\log \beta)$ in the expected query time.  In Section~\ref{sec:dyntridatastructure}, $\beta = O(\log n)$.  The fragmentation is higher now with $\beta = O(\log^2 n)$, but it still leads to the same additive term of $O(\log\log n)$ in the expected query time.
	
	\item  Given a query point, if we follow a pointer from a blocked gap $g$ of a node $v$ of $\cal T$ to a gap tree $T'$ at a child $w$ of $v$, we collect two terms in the query time, namely $\log \frac{\text{weight}(T)}{\wt(g)} + \log \frac{\text{weight}(T')}{\wt(g')}$, where $T$ is the gap tree at $v$ that contains $g$, and $g'$ is the gap in $T'$ that contains the query point.  
	
	\hspace*{12pt}In Section~\ref{sec:dyntridatastructure}, $\wt(g)\geq \text{weight}(T')$ because the gaps in $T'$ form a subset of $g$, and the weight of a gap is equal to the probability of the query point falling inside it.  Therefore, $\wt(g)$ and $\text{weight}(T')$ cancel, giving the term $\log \frac{\text{weight}(T)}{\wt(g')}$.   Such pairwise cancellation results in a bound of $O\bigl(\log\frac{1}{\wt(g'')}\bigr)$ on the query time, where $g''$ is the gap  of the last node visited that contains the query point.
	
	\hspace*{12pt}In our case here, some explanation is needed
	because the gap weight definition has been changed.  Consider a triangle $t$ in $\Delta_B$ that overlaps with $g$.    The triangle $t$ does not span $v$ as $g$ is a blocked gap.  So the contribution of $t$ to $\wt(g)$ is $b \cdot \max\{\Pr(t), 1/N\}$, where $b$ is the number of vertices of $t$ in $g$.  At a child $w$ of $v$, $t \cap \slab(w)$ is either a free gap or part of a blocked gap.  In the former case, the weight of $t \cap \slab(w)$ is $\max\{\Pr(t), 1/N\}$.  In the latter case, the weight contribution of $t$ at $w$ is at most $b \cdot \max\{\Pr(t), 1/N\}$.  It follows that $\text{weight}(T') \leq \wt(g)$.  
	Hence, we still have the pairwise cancellation.
	
	\cancel{
	\hspace*{12pt}Let $\tilde{g}$ be the gap at which the querying terminates.  The total time spent on searching gap trees telescopes to $O\bigl(\log \frac{1}{\wt(\tilde{g})}\bigr)$.  If $\tilde{g}$ is a blocked gap, $\cal T$ does not answer the query successfully.  So the above term is $\Theta(\log n)$ which is bounded by the $O(\log n)$ query time of the worst-case dynamic planar point location structure that will be queried next.  If $\tilde{g}$ is a free gap, then
	$O\bigl(\log \frac{1}{\wt(\tilde{g})}\bigr) = O\bigl(\log \frac{\log n}{\Pr(t)}\bigr) = O\bigl(\log \frac{1}{\Pr(t)} + \log\log n \bigr) $, where $t$ is the triangle in $\Delta_B$ that contains $\tilde{g}$.  The term $O\bigl(\log \frac{1}{\Pr(t)} \bigr)$ has been accounted for in the proof of Lemma~\ref{lem:dyntrianalysis}.  The extra $O(\log\log n)$ term is absorbed by the expected query time bound of $O(\text{OPT} +\log\log n)$.
	}

	\item In Section~\ref{sec:dyntridatastructure}, after locating a blocked gap $g$ in $\gaplist(u)$ for some internal node $u$ of $\cal T$, we check the $O(1)$ pointers kept at $g$ to access a gap tree of the child $v$ of $u$ such that $\slab(v)$ contains the query point.  This takes $O(1)$ time.  Now, $g$ is associated with a biased search tree $T_g$ of pointers to $O(\log n)$ gap trees, one per child of $u$.  Recall that the weight of the node in $T_g$ for the child $v$ of $u$ is $\max\bigl\{\frac{\Pr(\slab(u))}{2\log_2 N + 2} , \Pr(\slab(v))\bigr\}$.  Therefore, descending from $g$ to a gap tree of a child $v$ now takes $O\bigl(\log \frac{\Pr(\slab(u))}{\Pr(\slab(v))} + 1\bigr)$ time.   Summing over the entire tree path in $\cal T$ traversed, we obtain $O\left(-\log \Pr(\slab(w)) + \text{depth}(w)\right)$, where $w$ is the last node of $\cal T$ visited.  We have argued previously that $\text{depth}(w) = O(-\log \Pr(\slab(w)))$.  Let $\tau$ be the triangle in $\Delta_B$ that contains the query point.  We have $-\log\Pr(\slab(w)) \leq -\log \Pr(\slab(w) \cap \tau)$, and it has been shown in the proof of Lemma~\ref{lem:dyntrianalysis} that the expected value of $-\log \Pr(\slab(w) \cap \tau)$ is $O(\text{OPT} + \log\log n)$.
	
\end{itemize}
\end{proof}

\section{Conclusion}


We presented a dynamic data structure for distribution-sensitive point location in convex subdivisions.  The expected query time is asympototically optimal in the comparison-based model.   The update time is $O(k\log^5 n)$ amortized, where $k$ is the update size.  It can be used in a randomized incremental construction of a Voronoi diagram of $n$ points such that nearest neighbor can be found in optimal expected query time in any intermediate Voronoi diagram, where the expectation is taken over the underlying query distribution.  The total expected construction time is $O(n\log^5 n)$, where the expectation is taken over a uniform distribution of permutations of the $n$ input points.   

There are several future research directions.  First, can the update time be improved further?  Second, can one design a dynamic data structure for connnected subdivisions that achieves the optimal expected query time in the distribution-sensitive setting?  Another research problem is to remove the assumption of a fixed, underlying query distribution.  One possiblity is to accommodate local changes in the query distribution.  Another possibility is to maintain historical queries for defining the probabilities of a query point falling into the sudivision faces.

\bibliographystyle{plainurl}

\bibliography{socgref}

\begin{thebibliography}{10}

\bibitem{paper:Adamy1998}
U.~Adamy and R.~Seidel.
\newblock On the exact worst case query complexity of planar point location.
\newblock {\em Journal of Algorithms}, 27(1):189--217, 2000.

\bibitem{paper:Afshani2015}
P.~Afshani, J.~Barbay, and T.~Chan.
\newblock Instance-optimal geometric algorithms.
\newblock {\em Journal of the ACM}, 64(1):3:1--3:38, 2017.

\bibitem{paper:Arge2006}
L.~Arge, G.S. Brodal, and L~Georgiadis.
\newblock Improved dynamic planar point location.
\newblock In {\em Proceedings of the 47th Annual IEEE Symposium on Foundations
  of Computer Science}, pages 305--314, 2006.

\bibitem{paper:Arya2007b}
S.~Arya, T.~Malamatos, D.~Mount, and K.~Wong.
\newblock Optimal expected-case planar point location.
\newblock {\em SIAM Journal on Computing}, 37(2):584--610, 2007.

\bibitem{paper:Baumgarten1994}
H.~Baumgarten, H.~Jung, and K.~Mehlhorn.
\newblock Dynamic point location in general subdivisions.
\newblock {\em Journal of Algorithms}, 17(3):342--380, 1994.

\bibitem{paper:bent85}
S.W. Bent, D.D. Sleator, and R.E. Tarjan.
\newblock Biased search trees.
\newblock {\em SIAM Journal on Computing}, 14(3):545--568, 1985.

\bibitem{paper:Bose2013}
Prosenjit Bose, Luc Devroye, Karim Douieb, Vida Dujmovic, James King, and Pat
  Morin.
\newblock Odds-on trees, 2010.
\newblock \href {http://arxiv.org/abs/1002.1092} {\path{arXiv:1002.1092}}.

\bibitem{paper:Chan2015}
T.~Chan and Y.~Nekrich.
\newblock Towards an optimal method for dynamic planar point location.
\newblock {\em SIAM Journal on Computing}, 47(6):2337--2361, 2018.

\bibitem{paper:Cheng1992}
S.-W. Cheng and R.~Janardan.
\newblock New results on dynamic planar point location.
\newblock {\em SIAM Journal on Computing}, 21(5):972--999, 1992.

\bibitem{paper:cheng2017adaptive}
S.-W. Cheng and M.-K. Lau.
\newblock Adaptive planar point location.
\newblock In {\em Proceedings of the 33rd International Symposium of
  Computational Geometry}, pages 30:1--30:15, 2017.

\bibitem{cheng2015adaptive}
S.-W. Cheng and M.-K. Lau.
\newblock Adaptive point location in planar convex subdivisions.
\newblock {\em International Journal of Computational Geometry and
  Applications}, 27(1--2):3--12, 2017.

\bibitem{CL18}
S.-W. Cheng and M.-K. Lau.
\newblock Adaptive planar point location, 2018.
\newblock \href {http://arxiv.org/abs/1810.00715} {\path{arXiv:1810.00715}}.

\bibitem{CPT96}
Y.-J. Chiang, F.P. Preparata, and R.~Tamassia.
\newblock A unified approach to dynamic point location, ray shooting, and
  shortest paths in planar maps.
\newblock {\em SIAM Journal on Computing}, 25(1):207--233, 1996.

\bibitem{paper:Chiang1992}
Y.-J. Chiang and R.~Tamassia.
\newblock Dynamization of the trapezoid method for planar point location in
  monotone subdivisions.
\newblock {\em Internatational Journal of Computational Geometry and
  Applications}, 2(3):311--333, 1992.

\bibitem{paper:Collette2012}
S.~Collette, V.~Dujmovi\'{c}, J.~Iacono, S.~Langerman, and P.~Morin.
\newblock Entropy, triangulation, and point location in planar subdivisions.
\newblock {\em ACM Transactions on Algorithms}, 8(3):29:1--29:18, 2012.

\bibitem{Dobkin1990}
D.P. Dobkin and D.G. Kirkpatrick.
\newblock Determining the separation of preprocessed polyhedra---a unified
  approach.
\newblock In {\em Proceedings of the 17th International Colloquium on Automata,
  Languages and Programming}, pages 400--413, 1990.

\bibitem{driscoll89}
J.R. Driscoll, N.~Sarnak, D.D. Sleator, and R.E. Tarjan.
\newblock Making data structures persistent.
\newblock {\em Journal of Computer and System Sciences}, 38(1):86--124, 1989.

\bibitem{paper:Edelsbrunner1986}
H.~Edelsbrunner, L.~J Guibas, and J.~Stolfi.
\newblock Optimal point location in a monotone subdivision.
\newblock {\em SIAM Journal on Computing}, 15(2):317--340, 1986.

\bibitem{paper:Goodrich1997}
M.T. Goodrich and R.~Tamassia.
\newblock Dynamic ray shooting and shortest paths in planar subdivisions via
  balanced geodesic triangulations.
\newblock {\em Journal of Algorithms}, 23(1):51--73, 1997.

\bibitem{paper:Goodrich1998}
M.T. Goodrich and R.~Tamassia.
\newblock Dynamic trees and dynamic point location.
\newblock {\em SIAM Journal on Computing}, 28(2):612--636, 1998.

\bibitem{paper:Hershberger1993}
J.~Hershberger and S.~Suri.
\newblock A pedestrian approach to ray shooting: Shoot a ray, take a walk.
\newblock {\em Journal of Algorithms}, 18(3):403--431, 1995.

\bibitem{paper:Iacono2004}
J.~Iacono.
\newblock Expected asymptotically optimal planar point location.
\newblock {\em Computational Geometry: Theory and Applications}, 29(1):19--22,
  2004.

\bibitem{paper:Iacono2011}
J.~Iacono and W.~Mulzer.
\newblock A static optimality transformation with applications to planar point
  location.
\newblock {\em International Journal of Computational Geometry and
  Applications}, 22(4):327--340, 2012.

\bibitem{paper:Kirkpatrick1981}
D.~G. Kirkpatrick.
\newblock Optimal search in planar subdivisions.
\newblock {\em SIAM Journal on Computing}, 12(1):28--35, 1983.

\bibitem{O18}
E.~Oh.
\newblock Point location in incremental planar subdivisions.
\newblock In {\em Proceedings of the 29th International Symposium on Algorithms
  and Computation}, pages 51:1--51:12, 2018.

\bibitem{OA18}
E.~Oh and H.-K. Ahn.
\newblock Point location in dynamic planar subdivision.
\newblock In {\em Proceedings of the 34th International Symposium on
  Computational Geometry}, pages 63:1--53:14, 2018.

\bibitem{paper:Preparata1989}
F.P. Preparata and R.~Tamassia.
\newblock Fully dynamic point location in a monotone subdivision.
\newblock {\em SIAM Journal on Computing}, 18(4):811--830, 1989.

\bibitem{paper:Sarnak1986}
N.~Sarnak and R.~E. Tarjan.
\newblock Planar point location using persistent search trees.
\newblock {\em Communications of ACM}, 29(7):669--679, 1986.

\bibitem{paper:Shannon2001}
C.~E. Shannon.
\newblock A mathematical theory of communication.
\newblock {\em Mobile Computing and Communications Review}, 5(1):3--55, 2001.

\end{thebibliography}

\appendix

\section{Correctness of {\sc Modify}}
\label{app:modify}

We use $L^i_v$ to denote the current $L_v$ at the end of the $i$-th iteration of steps~3--6 of {\sc Modify}$(v)$.  $L^0_v$ is the initial $L_v$.  We use $L^*_v$ to refer to the list of free and blocked gaps of $v$ with respect to the triangulation after the triangulation-update.  Our goal is to show that $L^m_v = L^*_v$.

We imagine an incremental change in the underlying triangulation corresponding to the transformation from $L^{i-1}_v$ to $L^i_v$.  Specifically, at the end of the $(i-1)$-th iteration of steps~3--6 of {\sc Modify}$(v)$ in processing the components $C_1,\ldots,C_{i-1}$ in $\widetilde{H}_v$, we imagine that the portions of the triangulation covered by $C_1,\ldots,C_{i-1}$ have been updated to $T_U \cap C_1,\ldots,T_U \cap C_{i-1}$, respectively, whereas the portions covered by $C_i,\ldots,C_m$ are not updated yet.   We refer to the triangles of this intermediate triangulation that overlap with $\slab(v)$ as \emph{triangles with respect to $L^{i-1}_v$}.  Depending on the context, we may also be referring to a subset of these triangles.

For any simple curve $\xi$ in $\slab(v)$ that connects the two bounding lines of $\slab(v)$, let $L^i_v[\xi]$ and $L^*_v[\xi]$ denote the sublists of free and blocked gaps in $L^i_v$ and $L^*_v$, respectively, that lie entirely above $\xi$.   So any gap whose interior is intersected by $\xi$ is excluded.  We will prove by induction that the following invariants hold during the $i$-th iteration of steps 3--6 of {\sc Modify}$(v)$.   The gaps $\hat{g}_i$ and $\check{g}_i$ are defined as in {\sc Modify}$(v)$.

\begin{itemize}
	\item Inv1$(i)$: $L^{i-1}_v[\up(\hat{g}_i)] = L^*_v[\up(\hat{g}_i)]$.  
	
	\item Inv2$(i)$:  Define a gap $g$ in $L^*_v$ as follows: if $\up(\hat{g}_i)$ intersects the interior of a blocked gap in $L^*_v$, let $g$ be this blocked gap; otherwise, let $g$ be the highest gap in $L^*_v$ that lies below $\up(\hat{g}_i)$.  Note that $\up(\hat{g}_i)$ does not intersect the interior of any free gap in $L^*_v$.  At the end of step~4 of the $i$-th iteration of steps~3--6 of {\sc Modify}$(v)$, the following properties hold:
	\begin{itemize}
		\item If $g$ is a blocked gap and $\up(g)$ is above or partly above $\up(C_i)$, then $\mathit{state} = \mathsf{blocked}$, $g$ contains $\up(C_i)$, and $\gamma = \up(g)$.
		\item Otherwise, $\mathit{state} = \mathsf{null}$, $g$ is below $\up(C_i)$, and $\gamma = \up(C_i)$.
	\end{itemize}
	
	\item Inv3$(i)$: 
	\begin{itemize}
		\item If $\mathit{state} = \mathsf{null}$ at the beginning of step~6(b), then the last region in $C_i$ processed in step~5 spans $v$ and $L^i_v[\low(C_i)] = L^*_v[\low(C_i)]$.
		
		\item Otherwise, there is a gap $\tilde{g}$ in $L^*_v$ that contains $\low(C_i)$ and intersects $\reg(C_i)$.
		\begin{itemize}
			\item If $\mathit{state} = \mathsf{blocked}$ at the beginning of step~6(b) and $\mathit{state} = \mathsf{null}$ at the end of step~6(b), then both $\check{g}_i$ and $\tilde{g}$ are blocked gaps, $\low(\tilde{g}) = \low(\check{g}_i)$, and $L^i_v[\low(\check{g}_i)] = L^*_v[\low(\check{g}_i)]$.
		
			\item If $\mathit{state} = \mathsf{blocked}$ at the end of step~6(b), then $\gamma = \up(\tilde{g})$, both $\check{g}_i$ and $\tilde{g}$ are blocked gaps, both $\check{g}_i$ and $\tilde{g}$ contain $\up(C_{i+1})$, and $L^i_v[\low(C_i)] = L^*_v[\low(C_i)]$.
		\end{itemize}
	\end{itemize}
	%
\end{itemize}

Before proving the three invariants, we first establish a result showing that the upper and lower boundaries of the $C_r$'s, $\hat{g}_i$'s, $\check{g}_i$'s, and the gaps in $L^*_v$ do not cross each other.   That is, these upper and lower boundaries can be ordered from top to bottom within $\slab(v)$.  Therefore, if one boundary $\xi$ is above or partly above another boundary $\xi'$,  then $\xi$ precedes $\xi'$ in the top to bottom order as $\xi$ and $\xi'$ do not cross.

\begin{lemma}
	\label{lem:non-cross}
	No two upper or lower boundaries of any $C_r$, any $\hat{g}_i$, any $\check{g}_i$, or any gap in $L^*_v$ can cross each other.
\end{lemma}
\begin{proof}
	It is obvious that no two upper or lower boundaries of the $C_r$'s can cross each other.  The same holds among the upper and lower boundaries of gaps in $L^*_v$, as well as among the upper and lower boundaries of the $\hat{g}_i$'s and $\check{g}_i$'s.  Let's consider other combinations.
	
	Most  boundaries of the $\hat{g}_i$'s and $\check{g}_i$'s are gap boundaries before the triangulation-update, and therefore, they cannot cross the upper and lower boundaries of any $C_r$ by the definition of $T_U$.   Nevertheless, the upper boundary of $\hat{g}_j$ may be created in the $(j-1)$-th iteration of steps~3--6 of {\sc Modify}$(v)$ when $\check{g}_{j-1}$ is split at $\low(C_{j-1})$.  That is, $\up(\hat{g}_j)$ may be equal to $\low(C_{j-1})$.  Still, in this case, $\up(\hat{g}_j) = \low(C_{j-1})$ cannot cross the upper and lower boundary of any other $C_r$'s. 
	
	By Lemma~\ref{lem:ignore2}, every boundary of a gap in $L^*_v$ is contained in some $C_i$ or equal to a gap boundary before the triangulation-update.  It follows that a lower or upper boundary of a gap in $L^*_v$ cannot cross the lower or upper boundary of any $C_r$, any $\hat{g}_i$, or any $\check{g}_i$.
\end{proof}

Next, we deal with the base cases of the first two invariants, namely, Inv1(1) and Inv2(1).

\begin{lemma}
	\label{lem:inv11}
	\emph{Inv1}$(1)$ is true.
\end{lemma}
\begin{proof}
	By Lemma~\ref{lem:ignore}, the subset of $H_v \setminus \widetilde{H}_v$ above $\hat{g}_1$ can only overlap with blocked gaps in $L_v[\up(\hat{g}_1)]$ and $L^*_v$.   The triangulation-update does not affect the free gaps in $L_v[\up(\hat{g}_1)]$ and the portions of shadow gaps in $L_v$ above $\up(\hat{g}_1)$.
	
	Since blocked gaps are delimited by free or shadow gaps, we can almost conclude that the blocked gaps in $L_v[\up(\hat{g}_1)]$ are preserved in $L^*_v$, except for the possibility that there is a blocked gap $g \in L_v[\up(\hat{g}_1)]$ with $\low(g) = \up(\hat{g}_1)$.  The argument of being delimited by free or shadow gaps alone does not prove that $g$ is preserved in $L^*_v$ because $g$ is not delimited by anything from below in $L_v[\up(\hat{g}_1)]$.
	
	In this case, $\hat{g}_1$ is not a blocked gap because it is adjacent to $g$ in $L_v$.   We also know that $\up(C_1)$ is strictly below $\up(\hat{g}_1)$; otherwise, $\up(C_1) $ overlaps with $\low(g)$ and we should have taken $g$ to be $\hat{g}_1$ instead.  It implies that $\up(C_1)$ and $\reg(C_1)$ intersect the interior of $\hat{g}_1$, which forbids $\hat{g}_1$ from being a free gap.  So $\hat{g}_1$ is a shadow gap.  The region sandwiched between $\up(\hat{g}_1)$ and $\up(C_1)$ will remain a part of a shadow gap of $v$ after the triangulation-update.  This forbids the merging of $g$ with anything immediately below it after the triangulation-update.  Therefore, $g$ is preserved as a blocked gap in $L^*_v$.
	
	We conclude that $L_v[\up(\hat{g}_1)] = L^*_v[\up(\hat{g}_1)]$.
\end{proof}

Next, we prove Inv2$(i)$ in two steps.  First, we deal with the case of $\mathit{state} = \mathsf{null}$ at the beginning of step~4.  This will help us to establish Inv2(1).

\begin{lemma}
	\label{lem:inv2i0}
	If \emph{Inv1}$(i)$ is true and $\mathit{state} = \mathsf{null}$ at the beginning of step~4 in the $i$-th iteration of steps~3--6 of {\sc Modify}$(v)$, then \emph{Inv2}$(i)$ is true.
\end{lemma}
\begin{proof}
	Define the gap $g$ in $L^*_v$ as in Inv2$(i)$, namely, if $\up(\hat{g}_i)$ intersects the interior of a blocked gap in $L^*_v$, let $g$ be this blocked gap; otherwise, let $g$ be the highest gap in $L^*_v$ that lies below $\up(\hat{g}_i)$.  
	
	Recall that $L^*_v$ consists of free and blocked gaps only.   Since $\mathit{state} = \mathsf{null}$ at the beginning of step~4, either $i = 1$, or $\up(\hat{g}_i)$ was not created in step~6(b) of the $(i-1)$-th iteration by splitting $\check{g}_{i-1}$ at $\low(C_{i-1})$.  Therefore, the gaps below $\low(\check{g}_{i-1})$ were not changed in the $(i-1)$-th iteration.

	\begin{itemize}
		
		\item Case~1: $g$ is a blocked gap and $\up(g)$ is above or partly above $\up(C_i)$.  
		
		\hspace*{12pt}First, we claim that no part of $\up(\hat{g}_i)$ is above $\up(g)$.  If not, $\up(g)$ must lie at, below, or partly below $\up(C_i)$ by Lemmas~\ref{lem:ignore2} and~\ref{lem:non-cross}.   But this contradicts our assumption that $\up(g)$ is above or partly above $\up(C_i)$.
				
		\hspace*{12pt}Second, we claim that no part of $\up(g)$ is above $\up(\hat{g}_i)$.  Suppose not.  Since $g$ is a blocked gap, $\up(\hat{g}_i)$ intersects the interior of $g$ by the definition of $g$.  So $g \not\in L^*_v[\up(\hat{g}_i)]$.  All triangles that overlap with $g$ after the triangulation-update do not span $v$, including the subset $X$ that lie below $\up(g)$ and above $\up(\hat{g}_i)$.   The subset $X$ also exist with respect to $L^{i-1}_v$ because $\up(C_i)$ lies at, below, or partly below $\up(\hat{g}_i)$.   So the subset $X$ must overlap with a blocked gap $g'$ in $L^{i-1}_v[\up(\hat{g}_i)]$, implying that $g' \in L^*_v[\up(\hat{g}_i)]$ as $L^{i-1}_v[\up(\hat{g}_i)] = L^*_v[\up(\hat{g}_i)]$ by Inv1$(i)$.  However, $g' \not= g$ because $g \not\in  L^*_v[\up(\hat{g}_i)]$, contradicting the fact that gaps in $L^*_v$ are interior-disjoint.
	
		\hspace*{12pt}So $\up(g) = \up(\hat{g}_i)$ which is above or partly above $\up(C_i)$.  We claim that $\hat{g}_i$ is a blocked gap.  If not, $\hat{g}_i$ is a free or shadow gap, so all triangles with respect to $L^{i-1}_v$ that overlap with $\hat{g}_i$ span $v$, including those that lie below $\up(\hat{g}_i) = \up(g)$ and above $\up(C_i)$.  These triangles continue to exist after the triangulation-update, contradicting the assumption that $g$ is a blocked gap.  
		
		\hspace*{12pt}As a result, $\mathit{state} = \mathsf{null}$ before the beginning of Step~4, $\hat{g}_i$ is a blocked gap, and $\up(\hat{g}_i) = \up(g)$ is above or partly above $\up(C_i)$.  By these conditions, Step~4 sets $\mathit{state} := \mathsf{blocked}$ and $\gamma := \up(\hat{g}_i) = \up(g)$.  If $g$ does not contain $\up(C_i)$, some part of $\low(g)$ is above $\up(C_i)$.  As $g$ is a blocked gap in $L^*_v$, it is delimited by free and shadow gaps with respect to $L^*_v$.  Therefore, some triangles that lie below $\low(g)$ and above $\up(C_i)$ after the triangulation-update span $v$.  These triangles also exist with respect to $L^{i-1}_v$.  But then $\hat{g}_i$ must end above these triangles, contradicting the fact that $\up(C_i)$ intersects $\hat{g}_i$.

		\item Case~2: $g$ is not a blocked gap or no part of $\up(g)$ is above $\up(C_i)$.
		
		\hspace*{12pt}We claim that $\hat{g}_i$ is not a blocked gap or no part of $\up(\hat{g}_i)$ is above $\up(C_i)$.  Suppose not.  Then, $\hat{g}_i$ is a blocked gap and $\up(\hat{g}_i)$ is above or partly above $\up(C_i)$.  Note that $\low(\hat{g}_i)$ lies at, below, or partly below $\up(C_i)$.  If not, since $\up(C_i)$ intersects $\hat{g}_i$, Lemma~\ref{lem:non-cross} implies that $\up(C_i)$ overlaps partly with $\low(\hat{g}_i)$ and lies partly below $\low(\hat{g}_i)$.  So the gap in $L^{i-1}_v$ that lies immediately below $\hat{g}_i$ must be a blocked gap.  Recall that the gaps below $\low(\check{g}_{i-1})$ were not changed in the $(i-1)$-th iteration, which means that the blocked gap $\hat{g}_i$ and the blocked gap adjacent to and below $\hat{g}_i$ already exist before the triangulation-update, a contradiction to the fact that no two blocked gaps are adjacent before the triangulation-update.  As $\hat{g}_i$ is a blocked gap, the space below $\up(\hat{g}_i)$ and above $\up(C_i)$ is filled with triangles with respect to $L^{i-1}_v$ that do not span $v$.   These triangles continue to exist after the triangulation-update, meaning that there is a blocked gap $g'$ in $L^*_v$ that contains the portions of these triangles inside $\slab(v)$.  Observe that we must have chosen $g'$ to be $g$.  But then $g$ is a blocked gap and some part of $\up(g) = \up(g')$ is above $\up(C_i)$, contradicting the assumption of Case~2.
		
		\hspace*{12pt}Since $\mathit{state} = \mathsf{null}$ before the beginning of Step~4, our claim above implies that Step~4 does not alter \emph{state} and sets $\gamma := \up(C_i)$.
		
		\hspace*{12pt}We claim that $g$ is below $\up(C_i)$.  If $g$ is a blocked gap, then by the assumption of Case~2, $g$ is below $\up(C_i)$ as claimed.  Suppose that $g$ is a free gap.  Assume to the contrary that $\up(g)$ is above or partly above $\up(C_i)$.  By Lemma~\ref{lem:non-cross}, $\up(g)$ and $\up(C_i)$ do not cross.  Being a free gap, $g$ must then be above $\up(C_i)$ and, by the definition of~$g$, $\up(\hat{g}_i)$ must lie above $g$.  That is, $\up(\hat{g}_i)$ is strictly above $\up(C_i)$.  By the definition of $\hat{g}_i$, this happens only if $\hat{g}_i$ is a blocked gap.  The triangle containing $g$ spans $v$, and this triangle also exists with respect to $L^{i-1}_v$.  But then $\hat{g}_i$ is separated from $\up(C_i)$ by this triangle, contradicting the definition of $\hat{g}_i$.
\end{itemize}
\end{proof}

\begin{corollary}
	\label{cor:inv21}
	\emph{Inv2}$(1)$ is true.
\end{corollary}	
\begin{proof}
	Inv1$(1)$ holds by Lemma~\ref{lem:inv11}.  Then, the result follows from Lemma~\ref{lem:inv2i0}.  
\end{proof}

The next result shows that Inv2$(i)$ holds for $i \geq 2$ inductively given Inv3$(i-1)$ and Inv1$(i)$.

\begin{lemma}
	\label{lem:inv2i}
	$\forall\, i \geq 2$, \emph{Inv3}$(i-1)$ $\wedge$ \emph{Inv1}$(i)$ $\Rightarrow$ \emph{Inv2}$(i)$.
\end{lemma}
\begin{proof}
	Lemma~\ref{lem:inv2i0} already establishes that Inv2$(i)$ is true in the case of $\mathit{state} = \mathsf{null}$ at the beginning of step~4 in the $i$-th iteration of steps~3--6 of {\sc Modify}$(v)$.  Assume that $\mathit{state} = \mathsf{blocked}$ at the beginning of step~4.  Note that step~4 does not alter \emph{state} in this case.  So $\mathit{state} = \mathsf{blocked}$ at the end of step~4.  It also means that $\mathit{state} = \mathsf{blocked}$ at the end of step~6(b) in the $(i-1)$-th iteration.  By Inv3$(i-1)$, 
	\begin{itemize}
		\item $\check{g}_{i-1}$ is a blocked gap and $\check{g}_{i-1}$ contains both $\low(C_{i-1})$ and $\up(C_i)$,
		\item there is a blocked gap $\tilde{g}$ in $L^*_v$ that contains both $\low(C_{i-1})$ and $\up(C_i)$, and
		\item $\gamma = \up(\tilde{g})$.
	\end{itemize}


	Since $\check{g}_{i-1}$ contains $\up(C_i)$, by the definition of $\check{g}_{i-1}$, either $\up(\check{g}_{i-1}) = \low(C_{i-1})$ or $\low(C_{i-1})$ intersects the interior of $\check{g}_{i-1}$.  Thus, the portion of $\check{g}_{i-1}$ below $\low(C_{i-1})$ appears as a blocked gap in $L^{i-1}_v$ by step~6(b) of the $(i-1)$-th iteration, which implies that $\hat{g}_i$ is the portion of  $\check{g}_{i-1}$ below $\low(C_{i-1})$.   So $\up(\hat{g}_i) = \low(C_{i-1})$.  Since $\tilde{g}$ is a blocked gap and $\tilde{g}$ contains both $\low(C_{i-1}) = \up(\hat{g}_i)$ and $\up(C_i)$, either $\up(\tilde{g}) = \up(\hat{g}_i)$ or $\up(\hat{g}_i)$ intersects the interior of $\tilde{g}$.  In either case, we must have chosen $\tilde{g}$ to be the gap $g$ defined in Inv2$(i)$.  So Inv2$(i)$ is satisfied.
\end{proof}

The next result shows that Inv3$(i)$ is true given Inv1$(i)$ and Inv2$(i)$.

\begin{lemma}
	\label{lem:inv3i}
	$\forall\, i \geq 1$, \emph{Inv1}$(i)$ $\wedge$ \emph{Inv2}$(i)$ $\Rightarrow$ \emph{Inv3}$(i)$.
\end{lemma}
\begin{proof}
	By Inv2$(i)$ and the specification of step~5 of {\sc Modify}$(v)$, for each region $\tau$ in $C_i$ encountered in an iteration of step~5, the following properties are satisfied at the end of that iteration.
	\begin{quote}
		{\sc Properties of step 5:}
		\begin{romani}
			\item if $\tau$ spans $v$, then $\Sigma = L^*_v[\low(\tau)] \setminus L^*_v[\up(\hat{g}_i)]$;
			\item if $\tau$ is a legal region that does not span $v$, then $\tau$ is contained in a blocked gap $\tilde{g}$ in $L^*_v$, $\Sigma = L^*_v[\up(\tilde{g})] \setminus L^*_v[\up(\hat{g}_i)]$, and $\gamma = \up(\tilde{g})$.
		\end{romani}
	\end{quote}
	Our proof distinguishes three cases according to the value of $\mathit{state}$ as in the description of Inv3$(i)$.
	\begin{itemize}
		\item $\mathit{state} = \mathsf{null}$ at the beginning of step~6(b).  
		
		\hspace*{12pt}Let $\tau$ be the last region in $C_i$ encountered in step~5.  As $\mathit{state} = \mathsf{null}$ after step~5, $\tau$ must span $v$, which implies $\low(\tau) = \low(C_i)$.  By property~(i) of step~5 as mentioned above, we have $\Sigma = L^*_v[\low(\tau)] \setminus L^*_v[\up(\hat{g}_i)] = L^*_v[\low(C_i)] \setminus L^*_v[\up(\hat{g}_i)]$.   By Inv1$(i)$, $L^{i-1}_v[\up(\hat{g}_i)] = L^*_v[\up(\hat{g}_i)]$.  Therefore, after step~6 replaces by $\Sigma$ the gaps in $L^{i-1}_v$ from $\hat{g}_i$ to $\low(C_i)$, we obtain $L^i_v[\low(C_i)] = L^*_v[\low(C_i)]$.
		
		\item $\mathit{state} = \mathsf{blocked}$ at the beginning of step~6(b).  
		
		\hspace*{12pt}Let $\tau$ be the last region in $C_i$ encountered in step~5.  With a slight abuse of notation, although $\tau$ does not span $v$, we also use $\up(\tau)$ and $\low(\tau)$ to mean the upper and lower boundaries of $\tau$, respectively.  In this case, $\tau$ is a legal region that does not span $v$.  Note that $\low(\tau) \subset \low(C_i)$.  
		
		\hspace*{12pt}We claim that $\check{g}_i$ is a blocked gap.  Since $\low(\tau) \subset \low(C_i)$, there are at least two edges in $\low(C_i)$.  Let $x$ be an interior vertex of $\low(C_i)$.   The edges of $\low(C_i)$ subtend two angles at $x$.  
		Either the angle above $\low(C_i)$ at $x$ or the angle below $\low(C_i)$ at $x$ is at least $\pi$.  There must then be a triangle $t$ with respect to $L^{i-1}_v$ that is incident to $x$ and does not span $v$.   Therefore, $\low(C_i)$ intersects the blocked gap in $L^{i-1}_v$ that overlaps with $t$, implying that $\check{g}_i$ is a blocked gap.
		
		\hspace*{12pt}Since $\tau$ is a legal region that does not span $v$, by property (ii) of step~5 mentioned above, there is a blocked gap $\tilde{g}$ in $L^*_v$ that contains $\tau$.  The existence of $\tau$ means that all regions in $C_i$ that are incident to and above $\low(C_i)$ do not span $v$.  It follows that $\tilde{g}$ contains $\low(C_i)$ and intersects $\reg(C_i)$, including the possibility of $\low(\tilde{g}) = \low(C_i)$, such that $\Sigma = L^*_v[\up(\tilde{g})] \setminus L^*_v[\up(\hat{g}_i)]$ and $\gamma = \up(\tilde{g})$ at the end of step~5.  There are two subcases depending on the value of $\mathit{state}$ at the end of step~6(b).
		\begin{itemize}
			\item $\mathit{state} = \mathsf{null}$ at the end of step~6(b).   
			
			\hspace*{12pt}It means that either $C_i$ is the last component in $\widetilde{H}_v$ or $\check{g}_i$ is disjoint from $\up(C_{i+1})$.  In either case, the triangulation-update does not affect the triangles with respect to $L^{i-1}_v$ that lie below $\low(C_i)$ and overlap with the blocked gap $\check{g}_i$.  Therefore, the blocked gap $\tilde{g}$ in $L^*_v$ ends at $\low(\check{g}_i)$, i.e., $\low(\tilde{g}) = \low(\check{g}_i)$.
			
			\hspace*{12pt}Step~6(b) appends to $\Sigma$ a blocked gap $g$ with $\up(g) = \gamma = \up(\tilde{g})$ and $\low(g) = \low(\check{g}_i)$.  By property (ii) of step~5 mentioned above, this addition makes $\Sigma = L^*_v[\low(\check{g}_i)] \setminus L^*_v[\up(\hat{g}_i)]$.  By Inv1(i), $L^{i-1}_v[\up(\hat{g}_i)] = L^*_v[\up(\hat{g}_i)]$.  Then, after step~6(b) replaces by $\Sigma$ the gaps in $L^{i-1}_v$ from $\hat{g}_i$ to $\check{g}_i$, we obtain $L^i_v[\low(\check{g}_i)] = L^*_v[\low(\check{g}_i)]$.
			
			\item $\mathit{state} = \mathsf{blocked}$ at the end of step~6(b).  
			
			\hspace*{12pt}In this case, $C_{i+1}$ exists and $\up(C_{i+1})$ intersects $\check{g}_i$.
			
			\hspace*{12pt}Step~6(b) splits $\check{g}_i$ in $L^{i-1}_v$ into two gaps at $\low(C_i)$, say $g$ above $\low(C_i)$ and $g'$ below $\low(C_i)$.  As $\check{g}_i$ is a blocked gap, the triangles with respect to $L^{i-1}_v$ that form $g'$ do not span $v$.   Moreover, among the triangles with respect to $L^{i-1}_v$ that form $g'$, the triangulation-update does not affect the subset $X$ that  lie between $\low(C_i)$ and $\up(C_{i+1})$.  Recall that $\tilde{g}$ is a blocked gap in $L^*_v$ that contains $\low(C_i)$ and intersects $\reg(C_i)$.   Therefore, $\tilde{g}$ must include the subset $X$, implying that $\low(C_i)$ intersects the interior of $\tilde{g}$.
			
			\hspace*{12pt}Then, step~6(b) replaces by $\Sigma$ the gaps in $L^{i-1}_v$ from $\hat{g}_i$ to $\low(C_i)$.   By Inv1$(i)$, $L^{i-1}_v[\up(\hat{g}_i)] = L^*_v[\up(\hat{g}_i)]$.  Then, by the fact that $\low(C_i)$ intersects the interior of $\tilde{g}$ and $\Sigma = L^*_v[\up(\tilde{g})] \setminus L^*_v[\up(\hat{g}_i)]$, the replacement by $\Sigma$ gives $L^i_v[\low(C_i)] = L^*_v[\up(\tilde{g})] = L^*_v[\low(C_i)]$.
		\end{itemize}
	\end{itemize}
\end{proof}

Lemma~\ref{lem:inv11}, Corollary~\ref{cor:inv21}, and Lemma~\ref{lem:inv3i} establish Inv1(1), Inv2(1), and Inv3(1).  The next result shows that Inv3$(i)$ implies Inv1$(i+1)$.  Then, Lemma~\ref{lem:inv2i} gives Inv2$(i+1)$ and Lemma~\ref{lem:inv3i} gives Inv3$(i+1)$, which completes the inductive argument.

\begin{lemma}
	$\forall\, i \geq 1$, \emph{Inv3}$(i)$ $\Rightarrow$ \emph{Inv1}$(i+1)$.
\end{lemma}
\begin{proof}
	There are three cases depending on the value of $\mathit{state}$ in the $i$-th iteration of steps~3--6 of {\sc Modify}$(v)$.  
	\begin{itemize}
		\item Case~1: $\mathit{state} = \mathsf{null}$ at the beginning of step~6(b) in the $i$-th iteration.  By Inv3$(i)$, $L^i_v[\low(C_i)] = L^*_v[\low(C_i)]$ and the last region in $C_i$ processed in step~5 spans $v$.  Therefore, with respect to $L^i_v[\low(C_i)] = L^*_v[\low(C_i)]$, the gap that is immediately above $\low(C_i)$ is a free or shadow gap.  We conclude that the gaps in $L^i_v$ and $L^*_v$ are either above or below $\low(C_i)$.    By the definition of $\hat{g}_{i+1}$, no part of $\up(C_{i+1})$ is above $\up(\hat{g}_{i+1})$.   Therefore, the gaps in $L^i_v$ that lie below $\low(C_i)$ and above $\up(\hat{g}_{i+1})$, if any, are unaffected by the triangulation-update.  This implies that $L^*_v$ inherits the gaps in $L^i_v[\up(\hat{g}_{i+1})] \setminus L^i_v[\low(C_i)]$.  Hence, $L^i_v[\up(\hat{g}_{i+1})] = L^*_v[\up(\hat{g}_{i+1})]$.
		
		\item Case~2: 	$\mathit{state} = \mathsf{blocked}$ at the beginning of step~6(b), and $\mathit{state} = \mathsf{null}$ at the end of step~6(b) in the $i$-th iteration.  By Inv3$(i)$, $L^i_v[\low(\check{g}_i)] = L^*_v[\low(\check{g}_i)]$.  From the specification of step~6(b), we know that $\check{g}_i$ is disjoint from $C_{i+1}$.  
		Since $\up(C_{i+1})$ lies at, below, or partly below $\up(\hat{g}_{i+1})$ by definition, the gaps in $L^i_v[\up(\hat{g}_{i+1})] \setminus L^i_v[\low(\check{g}_i)]$, if any, are preserved in $L^*_v$.  Hence, $L^i_v[\up(\hat{g}_{i+1})] = L^*_v[\up(\hat{g}_{i+1})]$.
		
		\item Case~3: $\mathit{state} = \mathsf{blocked}$ at the end of step~6(b).  
		By Inv3$(i)$, $\check{g}_i$ is a blocked gap that contains $\up(C_{i+1})$.   From the specification of step~6(b),  
		$\check{g}_i$ is split at $\low(C_i)$, and so the portion of $\check{g}_i$ below $\low(C_i)$ appears as a blocked gap in $L^i_v$.  This implies that $\hat{g}_{i+1}$ is the portion of $\check{g}_i$ below $\low(C_i)$, and so $\up(\hat{g}_{i+1}) = \low(C_i)$.  Inv3$(i)$ implies that $L^i_v[\up(\hat{g}_{i+1})] = L^i_v[\low(C_i)] = L^*_v[\low(C_i)] = L^*_v[\up(\hat{g}_{i+1})]$.
		
	\end{itemize}
\end{proof}

We use Inv3$(m)$ to prove the correctness of {\sc Modify}$(v)$.

\begin{lemma}
	\label{lem:modify}
	{\sc Modify}$(v)$ updates $\gaplist(v)$ correctly.  
\end{lemma}
\begin{proof}
	Consider the processing of the last component $C_m$ in $\widetilde{H}_v$.  There are two cases:
	
	\begin{itemize}
		\item $\mathit{state} = \mathsf{null}$ at the beginning of step~6(b), or 
		\item $\mathit{state} = \mathsf{blocked}$ at the beginning of step~6(b), and $\mathit{state} = \mathsf{null}$ at the end of step~6(b).  
	\end{itemize}
	
	In the first case, Inv3$(m)$ implies that $L^m_v[\low(C_m)] = L^*_v[\low(C_m)]$ and, with respect to $L^m_v$, there is a free or shadow gap that is immediately above $\low(C_m)$,  Therefore, gaps in $L^*_v$ are either above or below $\low(C_m)$.  The gaps below $\low(C_m)$ in $L^m_v$ are formed by  triangles that survive the triangulation-update.  These gaps are preserved from $L^m_v$ to $L^*_v$.
	
	In the second case, by Inv3$(m)$, $\check{g}_m$ is a blocked gap and there is a blocked gap $\tilde{g}$ in $L^*_v$ such that $\low(\tilde{g}) = \low(\check{g}_m)$ and $L^m_v[\low(\check{g}_m)] = L^*_v[\low(\check{g}_m)]$.  The gaps below $\low(\check{g}_m)$ in $L^m_v$ are preserved from $L^m_v$ to $L^*_v$ because they are formed by triangles that are unaffected by the triangulation-update.
	
	In summary, we conclude that $L^m_v = L^*_v$.
	%
\end{proof}

\end{document}